\documentclass[12pt]{article}
\usepackage{amsmath, amsthm, amssymb, amsfonts,bigints}
\usepackage[longnamesfirst]{natbib}
\usepackage{setspace}
\usepackage{graphicx}
\usepackage{caption,subcaption}
\usepackage{color}
\usepackage{xcolor,mathtools}
\usepackage{afterpage,array,tabularx,makecell,multirow}
\usepackage{appendix}
\usepackage{float}
\usepackage{epstopdf}
\usepackage{placeins,enumitem}
\usepackage[margin=1.25in]{geometry}
\usepackage{booktabs}
\usepackage{threeparttable}
\usepackage{units}
\usepackage{float}
\usepackage{appendix}
\floatstyle{plaintop}
\restylefloat{table}
\restylefloat{figure}
\usepackage{titlesec}
\usepackage[bottom]{footmisc}

\newif\ifkeepremark

\newcounter{subassumption}[asu]

\makeatletter
\renewcommand{\p@subassumption}{\theasu}
\makeatother

\newif \ifshowexplanations
\showexplanationstrue

\newtheorem{theorem}{Theorem}

\newtheorem{corollary}{Corollary}

\newtheorem{example}{Example}

\newtheorem{lemma}{Lemma}

\theoremstyle{theorem}

\theoremstyle{definition}

\theoremstyle{definition}
\newtheorem{assumption}{Assumption}

\usepackage{chngcntr}
\usepackage{apptools}

\let\oldthebibliography\thebibliography
\renewcommand{\thebibliography}[1]{%
  \oldthebibliography{#1}%
  \setlength{\itemsep}{0pt}
  \setlength{\parskip}{0pt}
}

\theoremstyle{plain} 
\newtheorem*{continuancex}{Example \continuanceref~(continued)}
\newenvironment{continuance}[1]
{\newcommand{\continuanceref}{\ref{#1}}\begin{continuancex}}
	{\end{continuancex}}

\newcommand{\p}{\mathbb{P}}
\newcommand{\E}{\mathbb{E}}
\newcommand{\V}{\mathrm{Var}}
\newcommand{\C}{\mathrm{Cov}}
\newcommand{\I}{\mathbb{I}}

\usepackage[colorlinks=true,
linkcolor=blue,
citecolor=blue,
urlcolor=blue]{hyperref}

\shortcites{angelopoulosPredictionpoweredInference2023,angelopoulosPPIEfficientPredictionPowered2023}
\defcitealias{battagliaInferenceRegressionVariables2025}{BCHS}
\defcitealias{dupasGenderDifferencesEconomics2026}{Dupas, Handlan, Modestino, et al.}

\begin{document}

\onehalfspacing 

\date{April 23, 2026}
\title{Bootstrapping with AI/ML-generated labels\thanks{%
We thank Jaap Abbring, Stephen Hansen, Christoph Rothe, and participants of the 2025 Royal Economic Society meetings in Birmingham. This material is based upon work supported by the National Science Foundation under Award No.~2521471 (Christensen). Gon\c{c}alves acknowledges support from the Social Sciences and Humanities Research Council of Canada grant RGPIN-435-2023-0352.  
}}
\author{Timothy Christensen\thanks{%
Department of Economics, Yale University.} 
\quad \quad
S\'{\i}lvia Gon\c{c}alves\thanks{%
Department of Economics, McGill University.}
\quad \quad
Benoit Perron\thanks{%
Department of Economics, Universit\'{e} de Montr\'{e}al, CIREQ, CIRANO.}}

\maketitle
\thispagestyle{empty}

\begin{abstract}  
\singlespacing
\noindent 
AI/ML methods are increasingly used in economics to generate binary variables (or labels) via classification algorithms. When these generated variables are included as covariates in regressions, even small misclassification errors can induce large biases in OLS estimators and invalidate standard inference. We study whether the bootstrap can correct this bias and deliver valid inference. We first show that a seemingly natural \emph{fixed-label bootstrap}, which generates data using estimated labels but relies on a corrupted version in estimation, is generally invalid unless a strong independence condition between the latent true labels and other covariates holds. We then propose a \emph{coupled-label bootstrap} that jointly resamples the true and imputed labels, and show it is valid without this condition. Two finite-sample adjustments further improve coverage: a variance correction for uncertainty in estimated misclassification rates and a Hessian rotation for near-singular designs. We illustrate the methods in simulations and apply them to investigate the relationship between wages and remote work status.

\vspace{8pt}
\end{abstract}

\clearpage
\setcounter{page}{1}

\pagenumbering{arabic}

\newpage

\section{Introduction}

It is now standard practice in economics and across the social sciences to generate new variables and measurements using artificial intelligence (AI) and machine learning (ML) methods. A leading use case is the generation of binary variables, or \emph{labels}, which are imputed using classification algorithms. To give just a few examples within economics, \cite{goldsmith-pinkhamGenderGapHousing2023} impute the gender of borrowers from borrower names, \cite{bursztynImmigrantNextDoor2024} impute the race of charitable donors from donor names, \cite{adams-prasslFirmConcentrationJob2023} and \cite{hansenRemoteWorkJobs2026} classify the flexibility of work arrangements from job postings metadata, and \citetalias{dupasGenderDifferencesEconomics2026}\ (\citeyear{dupasGenderDifferencesEconomics2026}) impute speaker gender and tone-of-voice from audio recordings of economics seminars. These imputed variables are rarely of interest themselves. Rather, they are used as inputs to econometric models for downstream inference tasks.

While the easy generation of labels has opened the door for new research, it is not without significant methodological challenges. \citeauthor{battagliaInferenceRegressionVariables2025} (\citeyear{battagliaInferenceRegressionVariables2025}, \citetalias{battagliaInferenceRegressionVariables2025} hereafter) show theoretically and empirically that even high performance classifiers with error rates below 1\% can generate substantial biases in parameter estimates due to misclassification error in the generated labels. Moreover, failure to correct for these biases leads to invalid inference on regression parameters. There is therefore a need for methods that correct bias and restore valid inference.

A central challenge in developing such methods is data constraints. A widely discussed approach for correcting bias requires large validation samples, in which both the imputed variables and ground-truth labels are observed alongside other variables in the model.\footnote{In economics, this approach goes back to \cite{boundExtentMeasurementError1991}, \cite{carrollSemiparametricEstimationLogistic1991}, \cite{sepanskiSemiparametricQuasilikelihoodVariance1993}, \cite{boundEvidenceValidityCrossSectional1994}, and \cite{leeEstimationLinearNonlinear1995}, among others, and has connections with the literature on auxiliary data \citep{chenMeasurementErrorModels2005,chenSemiparametricEfficiencyGMM2008}. Recent papers proposing these methods within the context of AI/ML-generated data in economics include \cite{ludwigLargeLanguageModels2025} and \cite{carlsonUnifyingFrameworkRobust2026}. Similar approaches are taken in the recent statistics literature on ``prediction-powered inference'' \citep{angelopoulosPredictionpoweredInference2023,angelopoulosPPIEfficientPredictionPowered2023}. See also \cite{fongMachineLearningPredictions2021} and \cite{egamiUsingLargeLanguage2024}, among others, for related work in political science.} In many economic applications, however, such data are unavailable. Instead, researchers may only observe an external sample from which the accuracy of the classifier can be measured. For instance, \cite{bursztynImmigrantNextDoor2024} estimate the accuracy of their classifier using an external sample of North Carolina voter registration data which contains self-reported ethnicity, but is missing data on charitable donations. Moreover, because collecting ground-truth labels is typically much more costly than generating labels with a classifier, the size $m$ of this external sample is usually much smaller than the sample size $n$ used for model estimation ($m \ll n$). \citetalias{battagliaInferenceRegressionVariables2025} propose analytic bias corrections that account for both these data constraints.

In this paper, we study whether the bootstrap can correct bias and deliver valid inference. There is a long tradition of using the bootstrap to reduce bias and improve the coverage of confidence intervals \citep{efronNonparametricStandardErrors1981,hallBootstrapEdgeworthExpansion1992}. In our setting, however, the problem is deceptively difficult. With binary labels, measurement error is necessarily ``nonclassical'' \citep{aignerRegressionBinaryIndependent1973}: a misclassified 1 must be a 0, and vice versa. Consequently, common intuitions and methods developed in settings with ``classical'' measurement error do not easily carry over.

More formally, consider the regression model
\[
 \begin{aligned}
 Y_i & = X_i' \beta + u_i , \\
 X_i & = g(\theta_i, Z_i),
 \end{aligned}
\]
where $\theta_i$ is a latent binary variable, $Y_i$ and $Z_i$ are observed, and the functional form of $g$ is known. For example, $g(\theta_i,Z_i) = (\theta_i, Z_i')'$ when $\theta_i$ enters the regression additively, or $g(\theta_i,Z_i) = (\theta_iZ_i', Z_i')'$ when it enters through interactions. In practice, researchers replace $\theta_i$ with an estimate $\hat \theta_i$, regress $Y_i$ on $\hat X_i \equiv g(\hat \theta_i, Z_i)$, and report confidence intervals and standard errors for $\beta$ treating $\hat X_i$ as though it is the truth. Misclassification ($\hat \theta_i \neq \theta_i$ for some observations) induces measurement error in the $\hat X_i$, so regression estimators can be biased and standard inference can be invalid.

As in \citetalias{battagliaInferenceRegressionVariables2025}, we assume the researcher observes a large sample $(Y_i, Z_i, \hat \theta_i)_{i=1}^n$ and a much smaller external sample $(\theta_i, \hat \theta_i)_{i=1}^m$ used to estimate classifier accuracy.\footnote{In fact, neither we nor \citetalias{battagliaInferenceRegressionVariables2025} need the full external sample, just estimates $\hat F_+$ and $\hat F_-$ of the false positive and negative rates of the classifier, and the sample size $m$ used to estimate those rates.} The question we address is: how should one bootstrap in this scenario?

Naively resampling $(Y_i,\hat X_i)$, or using equivalent residual-based approaches such as the wild bootstrap, does not work. These methods do not introduce any measurement error in the covariates in the bootstrap world, because the same covariates are used in data generation and parameter estimation. As a result, they fail to reproduce the bias of the OLS estimator and yield invalid inference.

A seemingly natural approach is a \emph{fixed-label bootstrap}, a variant of the wild bootstrap that builds in measurement error. Outcomes are generated as $Y_i^* = \hat X_i' \hat \beta + u_i^*$, where $\hat \beta$ is the OLS estimator from regression of $Y_i$ on $\hat X_i$, and $u_i^*$ is a bootstrap version of the OLS residual. Then, $Y_i^*$ is regressed on $\hat X_i^* \equiv g(\hat \theta_i^*, Z_i)$, where $\hat \theta_i^*$ is a corrupted version of $\hat \theta_i$ constructed so that the false positive and negative rates for predicting $\hat \theta_i$ with $\hat \theta_i^*$ match those observed in the external sample. Analogues of this bootstrap have been used to correct bias in other settings with generated covariates, including factor-augmented regressions \citep{goncalvesBootstrappingFactoraugmentedRegression2014,goncalvesBootstrapInferenceGroup2025} and panel data models \citep{goncalvesBootstrapInferenceLinear2015,higginsBootstrapInferenceFixedEffect2024,liConfidenceIntervalsTreatment2024}. However, we show that this bootstrap is generally inconsistent: it has the correct asymptotic variance but fails to replicate the asymptotic bias of the OLS estimator unless certain independence conditions hold between $\theta_i$ and $Z_i$. 

To address this problem, we introduce a novel \emph{coupled-label bootstrap}. Outcomes are generated as $Y_i^* = X_i^{*\prime} \hat \beta + u_i^*$, where $X_i^* = g(\theta_i^*,Z_i)$, $\theta_i^*$ is a bootstrap version of the true latent $\theta_i$, and $\hat \beta$ and $u_i^*$ are defined above. Estimation then proceeds by regressing $Y_i^*$ on $\hat X_i^*$ as before. The key innovation is to generate the pairs $(\theta_i^*, \hat \theta_i^*)$ jointly so that the misclassification rates for predicting $\theta_i^*$ with $\hat \theta_i^*$ match those in the external sample. We show that this bootstrap reproduces the correct asymptotic bias and variance of the OLS estimator, which justifies the use of confidence sets and hypothesis tests constructed via this bootstrap.

While our results establish asymptotic validity of the coupled-label bootstrap, two further modifications improve its finite-sample performance. The first is a variance correction that accounts for additional sampling error from estimating the misclassification rates. While \citetalias{battagliaInferenceRegressionVariables2025} develop analytic corrections, we show that a similar adjustment can be implemented via a simple modification of the bootstrap. The second is a rotation of the inverse Hessian, which is useful in near-singular designs where only a small fraction of $\theta_i$ are 0 (or 1). In these settings, misclassification can introduce discrepancies between the bootstrap and estimated inverse Hessian that are asymptotically negligible but material in finite samples. 
Simulations calibrated to our empirical application show that these modifications substantially improve finite-sample performance. We therefore recommend the coupled-label bootstrap with both  modifications.

The remainder of the paper is organized as follows. Section~\ref{sec:general} provides high-level conditions for bootstrap validity with generated covariates. Section~\ref{sec:labels} develops theory for the binary labels case, analyzing a general wild bootstrap and providing results for the fixed-label and coupled-label bootstraps. Section~\ref{sec:finite-sample} introduces two modifications that improve finite-sample performance. Section~\ref{sec:simulations} presents simulation evidence in a design calibrated to the application, and Section~\ref{sec:application} revisits \cite{hansenRemoteWorkJobs2026}, who study the relationship between wages and remote work status. All proofs are presented in the appendix.

\paragraph{Notation} Let $\p^*$ denote the bootstrap probability measure conditional on the data, and $\E^*$ the corresponding expectation. For bootstrap random vectors $Y_n^*$, we write $Y_n^* \to_{p^*} 0$ or $Y_n^* = o_{p^*}(1)$ if $\p^*(\|Y_n^*\| > \epsilon) \to_p 0$ for all $\epsilon > 0$. Similarly, $Y_n^* = O_{p^*}(1)$ if $\limsup_{n \to \infty} \p^*(\|Y_n^*\| > M) \to_p 0$ as $M \to \infty$. Finally, we write $Z_n^* \to_{d^*} Z$ to denote that $\p^*(Z_n^* \leq z) - F(z) \to_p 0$ for all continuity points $z$ of the CDF $F$ of $Z$.

\section{General Case}\label{sec:general}

In this section, we consider the general case in which $Y$ is regressed on a vector of generated covariates $\hat X$. We first recap results from \citetalias{battagliaInferenceRegressionVariables2025}, who show that regressing $Y$ on $\hat X$ induces a bias in the asymptotic distribution of the OLS estimator that invalidates standard inference. We then provide high-level conditions under which the bootstrap can match the asymptotic bias and variance of the OLS estimator, justifying the use of bootstrap percentile intervals to perform valid inference. The central question is then whether one can design a bootstrap satisfying these conditions in settings of interest. Section~\ref{sec:labels} addresses this question for the binary labels case. 

\subsection{Asymptotic Distribution of the OLS Estimator}

We wish to perform inference on the parameter vector $\beta$ in the regression model
\[
 Y_i = \beta' X_i + u_i,
\]
where the true covariates $X_i$ are latent and generated covariates $\hat X_i$ are used in their place. We are interested in the large-sample properties of the OLS estimator
\[
  \hat{\beta}=\Big(n^{-1}\sum_{i=1}^{n}\hat{X}_i\hat{X}^{\prime}_i\Big)^{-1}n^{-1}\sum_{i=1}^{n}\hat{X}_i Y_i
\]
from a regression of $Y_i$ on a vector of generated covariates $\hat X_i$. 

The following regularity conditions are similar to those in \citetalias{battagliaInferenceRegressionVariables2025}:

\begin{assumption}\label{a1suff1}
\begin{enumerate}[label={(\roman*)}]
\item\label{a1suff1.1} $\frac{1}{n} \sum_{i=1}^n \|\hat X_i - X_i\|^2 \to_p 0$;
\item\label{a1suff1.2} $\frac{1}{\sqrt n} \sum_{i=1}^n \hat X_i (\hat X_i - X_i)' \to_p B$;
\item\label{a1suff1.3} $\frac{1}{\sqrt n} \sum_{i=1}^n (\hat X_i - X_i) u_i \to_p 0$.
\end{enumerate}
\end{assumption}

\begin{assumption}\label{a1suff2}
\begin{enumerate}[label={(\roman*)}]
\item\label{a1suff2.1} $\frac{1}{n} \sum_{i=1}^n X_i X_i' \to_p Q \equiv \E[X_i X_i'] > 0$, for $\E[\|X_i\|^2] < \infty$;
\item\label{a1suff2.2} $\frac{1}{\sqrt n} \sum_{i=1}^n X_i u_i \to_d N(0,\Sigma)$, for $\Sigma = \E[ X_i X_i' u_i^2]$.
\end{enumerate}
\end{assumption}

Assumption~\ref{a1suff1} imposes conditions on the generated variable $\hat X_i$ while Assumption~\ref{a1suff2} is standard and only concerns the latent variable $X_i$. Assumption~\ref{a1suff1} allows for both classical and nonclassical error in $\hat X_i$. The key condition is Assumption~\ref{a1suff1}\ref{a1suff1.2}. In the classical case, this condition requires the measurement error variance to be $O(n^{-1/2})$, similar to the small measurement error asymptotic frameworks of \cite{chesherEffectMeasurementError1991} and \cite{evdokimovSimpleEstimationSemiparametric2023}. This scaling ensures that measurement error persists but does not dominate sampling error in the asymptotic distribution of $\hat \beta$. As a result, asymptotics deliver tractable and useful approximations to the finite-sample distribution of $\hat \beta$, where both sources of error play a role. Section~\ref{sec:labels} provides sufficient conditions for Assumption~\ref{a1suff1} for the binary labels case.

The next result follows by similar arguments to \citetalias{battagliaInferenceRegressionVariables2025} and so a proof is omitted:

\begin{theorem}\label{t1}
Suppose that Assumptions~\ref{a1suff1} and~\ref{a1suff2} hold. Then as $n \to \infty$,
\[
 \sqrt n \big(\hat \beta - \beta\big) \to_d N(b, V),
\]
where $b = - Q^{-1} B \beta$, and $V = Q^{-1} \Sigma Q^{-1}$.
\end{theorem}

As noted in \citetalias{battagliaInferenceRegressionVariables2025}, Theorem~\ref{t1} shows that if $b\ne 0$, then $\hat \beta$ has a first-order asymptotic bias and the usual approach to inference is invalid.

\subsection{Consistency of the Bootstrap}

Suppose that we generate the bootstrap data as
\[
  Y^*_i =\hat{\beta}^\prime X^*_i+u^*_i
\]
where $u^*_i$ is a bootstrap version of $\hat{u}_i=Y_i-\hat{\beta}^\prime \hat{X}_i$. For now, we do not specify a particular bootstrap method used in generating $u^*_i$. Later we will consider a wild bootstrap for i.i.d.~data. Similarly, $X^*_i$ is a bootstrap analogue of the true latent $X_i$. A natural choice for $X^*_i$ is $\hat{X}_i$, as in a standard fixed-regressor bootstrap. As we will see below, allowing for $X^*_i$ to differ from $\hat{X}_i$ turns out to be important for the binary labels application, so we allow for this possibility here.

Let $\hat{X}^*_i$ denote a bootstrap analogue of $\hat{X}_i$ which will be context-specific. The bootstrap analogue of $\hat{\beta}$ is
\begin{flalign*}
  \hat{\beta}^*=\Big(n^{-1}\sum_{i=1}^{n}\hat{X}^*_i\hat{X}^{*^\prime}_i\Big)^{-1}n^{-1}\sum_{i=1}^{n}\hat{X}^*_i Y^*_i.
\end{flalign*}
The following set of high-level conditions on $(u^*_i,X^*_i,\hat{X}^*_i)$ ensure the bootstrap distribution of $\sqrt{n}(\hat{\beta}^*-\hat{\beta})$ is consistent for the distribution of $\sqrt{n}(\hat{\beta}-\beta)$ given in Theorem~\ref{t1} as $n \to \infty$.
\setcounter{assumption}{0}
\renewcommand{\theassumption}{B.\arabic{assumption}}
\begin{assumption}[]\label{b1}
\begin{enumerate}[label={(\roman*)}]
   \item\label{b1.1}$\frac{1}{n}\sum_{i=1}^{n}\|\hat{X}^*_i-X^*_i\|^2\rightarrow_{p^*}0$ and $\frac{1}{n}\sum_{i=1}^{n}\|X^*_i-\hat{X}_i\|^2\rightarrow_{p^*}0$.
   \item\label{b1.2}$\frac{1}{\sqrt{n}}\sum_{i=1}^{n}\hat{X}^*_i(\hat{X}^*_i-X^*_i)^\prime \rightarrow_{p^*} B$.
   \item\label{b1.3}$\frac{1}{\sqrt{n}}\sum_{i=1}^{n}(\hat{X}^*_i-X^*_i)u^*_i\rightarrow_{p^*} 0$ and $\frac{1}{\sqrt{n}}\sum_{i=1}^{n}(X^*_i-\hat{X}_i)u^*_i\rightarrow_{p^*} 0$.
\end{enumerate}
\end{assumption}
\begin{assumption}\label{b2}
$\frac{1}{\sqrt{n}}\sum_{i=1}^{n}\hat{X}_iu^*_i \rightarrow_{d^*} N(0,\Sigma)$. 
\end{assumption}

Assumption~\ref{b1} is the bootstrap analogue of Assumption~\ref{a1suff1}. Part \ref{b1.1} is used to show that $n^{-1}\sum_{i=1}^{n}\hat{X}^*_i\hat{X}^{*\prime}_i \to_{p^*} Q$, given $n^{-1}\sum_{i=1}^{n}\hat{X}_i\hat{X}_i'\to_p Q$ under Assumptions~\ref{a1suff1} and \ref{a1suff2}. Part \ref{b1.3} and Assumption~\ref{b2} imply that $n^{-1/2}\sum_{i=1}^{n}\hat{X}^*_i u^*_i\to_{d^*} N(0,\Sigma)$. In particular, Assumption~\ref{b2} is the bootstrap analogue of Assumption~\ref{a1suff2}\ref{a1suff2.2} and follows by the application of a bootstrap CLT.  Assumption~\ref{b1} \ref{b1.2} is the crucial assumption that ensures that the bootstrap mimics the bias term $B$ in Assumption~\ref{a1suff1}\ref{a1suff1.2}. 
\begin{theorem}\label{Theorem:generalbootstraptheory}Suppose that Assumptions~\ref{a1suff1} and \ref{a1suff2} hold. If $(u^*_i,X^*_i,\hat{X}^*_i)$ satisfy Assumptions~\ref{b1} and \ref{b2}, then as $n\rightarrow\infty$,
\begin{flalign*}
  \sqrt{n}\big(\hat{\beta}^*-\hat{\beta}\big)\rightarrow_{d^*} N(b,V),
\end{flalign*}
where $b = - Q^{-1} B \beta$, and $V = Q^{-1} \Sigma Q^{-1}$.
\end{theorem}
Theorem~\ref{Theorem:generalbootstraptheory} shows that the bootstrap correctly replicates the limiting distribution of $\hat{\beta}$, including its bias and variance, thereby justifying the construction of bootstrap percentile confidence intervals for the elements of $\beta$. Specifically, a $100(1-\alpha)\%$ percentile interval for the $j$-th element $\beta_j$ of $\beta$ is given by
\[
[\hat{\beta}_j-\hat{q}_{1-\alpha/2,j},\hat{\beta}_j-\hat{q}_{\alpha/2,j}],
\]
where $\hat{q}_{\alpha,j}$ denotes the $\alpha$-quantile of the bootstrap distribution of $\hat{\beta}^*_j-\hat{\beta}_j$. Unlike asymptotic theory-based intervals, this interval requires neither an explicit bias correction nor a variance estimator.

\section{Binary Labels Case}\label{sec:labels}

We now turn to the binary-label case. We begin by describing the setup and specializing Theorem~\ref{t1} to this setting. We then consider two bootstrap methods, both of which are instances of the wild bootstrap. The first is a seemingly natural \emph{fixed-label bootstrap}, analogous to the standard fixed-regressor bootstrap, which is generally inconsistent unless strong independence conditions hold between $\theta_i$ and $Z_i$. We next present a \emph{coupled-label bootstrap} and show that it is consistent in settings where the fixed-label bootstrap fails. We therefore recommend the coupled-label bootstrap, together with the finite-sample modifications described in the next section.

\subsection{Setup}

In this setting, 
\[
 X_i = g(\theta_i, Z_i),
\]
where $g : \{0,1\} \times \mathbb R^{d_z} \to \mathbb R^k$ is known, $\theta_i$ is a latent binary random variable, and $Z_i \in \mathbb R^{d_z}$ is an observed random vector. Since $\theta_i$ is latent, it is common practice to generate an estimate $\hat \theta_i$ of $\theta_i$ using a classification algorithm, then regress $Y_i$ on $\hat X_i = g(\hat \theta_i, Z_i)$. For ease of exposition we will focus on the case of scalar $\theta_i$, though it is straightforward to extend our analysis to the vector case. 

\begin{example}\label{ex1}
This framework includes models where $\theta_i$ is interacted with a subset of the $Z_i$ variables, such as
\begin{equation}\label{eq:ex1}
 Y_i = \beta_1' \theta_i Z_{1i} + \beta_2' Z_{2i} + u_i,
\end{equation}
to capture, e.g., effects that vary by membership of a race or gender group.
In this case, $g(\theta_i, Z_i) = (\theta_i Z_{1i}', Z_{2i}')' \in \mathbb R^{k}$.
\end{example}

\begin{example}\label{ex:bchs}
Example~\ref{ex1} includes as a special case the model studied in \citetalias{battagliaInferenceRegressionVariables2025}, in which
\begin{equation}\label{eq:bchs}
 Y_i = \beta_1 \theta_i + \beta_2' Z_{2i} + u_i,
\end{equation}
where $Z_{1i} = 1$ and $Z_{2i} = Z_i$. For this model, $g(\theta_i,Z_i)=(\theta_i,Z_i^\prime)^\prime \in \mathbb R^{k}$.
\end{example}

To maintain flexibility, we will not take a stand on the algorithm with which $\hat \theta_i$ is generated. Instead, we shall assume that $(Y_i,Z_i,\hat \theta_i,\theta_i)_{i=1}^n$ are drawn i.i.d.~from a distribution $P_n$. Here we use the notation $P_n$ because our asymptotic framework introduced below allows the distribution of the data to vary with $n$, so that both misclassification and sampling error feature in the asymptotic distribution of the OLS estimator. The econometrician observes $(Y_i,Z_i, \hat \theta_i)_{i=1}^n$ and, for the purposes of bias correction, has access to an external sample $(\theta_i,\hat \theta_i)_{i=1}^m$, which can be used to assess the accuracy of the classifier. As in \citetalias{battagliaInferenceRegressionVariables2025}, we allow the external sample size $m$ to be of much smaller magnitude than $n$. This reflects common empirical settings with large amounts of generated data and small amounts of ground-truth observations.

\subsection{Asymptotic Bias of the OLS Estimator}

We now provide primitive conditions for Assumption~\ref{a1suff1} for the binary labels case. To this end we impose some structure on $\hat \theta_i$ and $\theta_i$. We follow \citetalias{battagliaInferenceRegressionVariables2025} and adopt an asymptotic framework where the (unconditional) false-positive rate $\E[\hat \theta_i(1-\theta_i)]$ and (unconditional) false-negative rate $\E[\theta_i(1-\hat \theta_i)]$ both tend to zero as $n \to \infty$, so that both sampling error in the regression and measurement error in the $\hat \theta_i$ remain relevant in large samples. This ensures the asymptotic distribution mimics the finite-sample setting where both play roles. 

To introduce the assumptions, note that we can write
\[
\hat{X}_i - X_i =\hat{\theta}_i(1-\theta_i)(g(1,Z_i) - g(0,Z_i)) +\theta_i(1-\hat{\theta}_i)(g(0,Z_i) - g(1,Z_i)),
\]
from which it follows that 
\[
\hat{X_i}(\hat{X}_i - X_i )'=\hat{\theta}_i(1-\theta_i)D_{+i} +\theta_i(1-\hat{\theta}_i)D_{-i},
\]
where 
\[
\begin{aligned}
 D_{+i} & = g(1,Z_i)(g(1,Z_i) - g(0,Z_i))' , \\
 D_{-i} & = g(0,Z_i)(g(0,Z_i) - g(1,Z_i))' .
\end{aligned}
\]
Let $G_i = \max\{\|g(0,Z_i)\|,\|g(1,Z_i)\|\}$ and let $\delta > 0$.

\renewcommand{\theassumption}{\arabic{assumption}}

\begin{assumption}\label{a2}
\begin{enumerate}[label={(\roman*)}]
\item\label{a2.1} $\E[X_i u_i] = 0$ and $\E[|u_i|^{4+\delta}] < \infty$;
\item\label{a2.2} $\E[X_i X_i'] > 0$ and $\E[G_i^{4+\delta}] < \infty$;
\item\label{a2.3} $\E[(\hat X_i - X_i)u_i] = o(n^{-1/2})$ as $n \to \infty$;
\item\label{a2.4} $\sqrt n \E[\hat \theta_i(1-\theta_i)] \to \kappa_+$ and $\sqrt n \E[\theta_i(1-\hat \theta_i)] \to \kappa_-$ as $n \to \infty$, for $0 \leq \kappa_+, \kappa_- < \infty$;
\item\label{a2.5} $\sqrt n ( \E [ \hat \theta_i(1-\theta_i) D_{+i} + \theta_i(1-\hat \theta_i) D_{-i} ] - \E[\hat \theta_i(1-\theta_i)] \E[D_{+i}] - \E[\theta_i(1-\hat \theta_i)] \E[D_{-i}]  ) \to 0$ as $n \to \infty$.
\end{enumerate}
\end{assumption}

Assumption~\ref{a2}\ref{a2.1}--\ref{a2.2} are standard regularity conditions. Assumption~\ref{a2}\ref{a2.3} allows for some correlation between the measurement error in $\hat X_i$ and $u_i$, provided it vanishes at a sufficiently fast rate. For instance, suppose $\theta_i$ and $\hat \theta_i$ are functions of $Z_i$ and some auxiliary information $\mathcal U_i$. Then by iterated expectations, 
\begin{flalign*}
 \E[(\hat{X}_i - X_i)u_i]
 &=\E\big[\hat{\theta}_i(1-\theta_i)(g(1,Z_i) - g(0,Z_i))\E[u_i|Z_i,\mathcal{U}_i]\big]\\
&\quad+\E\big[\theta_i(1-\hat{\theta}_i)(g(0,Z_i) - g(1,Z_i))\E[u_i|Z_i,\mathcal{U}_i]\big],
\end{flalign*}
so it suffices that $\E[u_i|Z_i,\mathcal{U}_i]=0$ for Assumption~\ref{a2.3} to be satisfied.

Assumption~\ref{a2}\ref{a2.4} formalizes the asymptotic framework described above. This condition requires the false positive and false negative rates to vanish at the same order as sampling uncertainty in the regression (i.e., $O(n^{-1/2})$). It provides a more robust starting point for empirical work than the standard approach of simply regressing $Y_i$ on $\hat X_i$ and performing OLS inference, which implicitly assumes $\kappa_+$ and $\kappa_-$ are zero. 

Finally, Assumption~\ref{a2}\ref{a2.5} limits the dependence between $g(0,Z_i)$ and $g(1,Z_i)$ and the classification errors in $\hat \theta_i$ (but not the dependence between $Z_i$ and $\theta_i$). This condition relates both to the functional form of $g$ and the statistical nature of the classification errors. This condition allows the conditional probabilities of misclassification $\E[\hat \theta_i(1-\theta_i)|Z_i]$ and $\E[\theta_i(1-\hat \theta_i)|Z_i]$ to depend on $Z_i$, provided the misclassification errors are asymptotically uncorrelated with $D_{+i}$ and $D_{-i}$. As this assumption is important, before proceeding we discuss some sufficient conditions for it.

\begin{continuance}{ex:bchs}
Consider (\ref{eq:bchs}) as in \citetalias{battagliaInferenceRegressionVariables2025}. Then
\[
 \begin{aligned}
 D_{+i}  = \begin{pmatrix}
  1 & 0 \\
  Z_i & 0 \end{pmatrix}, & & 
  D_{-i} & = \begin{pmatrix}
  0 & 0 \\
  -Z_i & 0 \end{pmatrix} ,
 \end{aligned}
\]
so a sufficient condition for Assumption~\ref{a2}\ref{a2.5} is
\begin{equation} \label{eq:suff.linear}
\begin{aligned}
 \sqrt n \E[\theta_i - \hat \theta_i] & \to 0, \\
 \sqrt n \E[(\theta_i - \hat \theta_i)Z_i] & \to 0.
\end{aligned}
\end{equation}
Suppose that $\hat \theta_i$ is generated by a logistic classifier with (potentially many) variables $W_{i,n}$, including $Z_i$ and a constant. That is,
\begin{equation}\label{eq:logistic}
 \p(\hat \theta_i = 1|W_{i,n}) = \frac{e^{W_{i,n}'\gamma_n}}{1+e^{W_{i,n}'\gamma_n}}.
\end{equation}
It is important to note that we are not assuming the logistic model is the DGP for $\theta_i$, just that it generates $\hat \theta_i$. Suppose the classifier has been pre-trained on a large external data set of $(\theta_i, W_{i,n})$. The parameters $\gamma_n$ of the classifier maximize the population likelihood
\[
 L(\gamma) = \E[ \theta_i W_{i,n}' \gamma - \log (1 + e^{W_{i,n}'\gamma})], 
\]
and therefore satisfy the first-order condition
\[
 0 = \E\left[ \left( \theta_i - \frac{e^{W_{i,n}'\gamma_n}}{1+e^{W_{i,n}'\gamma_n}} \right) W_{i,n} \right]. 
\]
Hence, by (\ref{eq:logistic}) and iterated expectations, we have 
\[
 \E[(\theta_i - \hat \theta_i)W_{i,n}] = \E[(\theta_i - \E[\hat \theta_i|W_{i,n}])W_{i,n}] = 0,
\]
from which it follows that (\ref{eq:suff.linear}) holds.
Note that the first condition in  (\ref{eq:suff.linear}) together with $\sqrt n \E[\hat \theta_i(1-\theta_i)] \to \kappa_+$ implies that $\sqrt n \E[\theta_i(1-\hat \theta_i)] \to \kappa_+$ also holds, and thus $\kappa_+ = \kappa_-$.
\end{continuance}

\begin{lemma}\label{l1}
Suppose that Assumption~\ref{a2} holds. Then Assumptions~\ref{a1suff1} and~\ref{a1suff2} hold with
\[
  B = \kappa_+ \E[D_{+i}]  + \kappa_- \E[D_{-i}] .
\]
\end{lemma}

Before proceeding,  we characterize the bias term $B$ in the context of the two running examples. 

\begin{continuance}{ex1}
Consider the interactions model (\ref{eq:ex1}). Then
\[
 B = \begin{pmatrix}
  \kappa_+ \E[Z_{1i}Z_{1i}'] & 0 \\
  (\kappa_+ - \kappa_-) \E[Z_{2i}Z_{1i}'] & 0 \end{pmatrix}.
\]
\end{continuance}

\begin{continuance}{ex:bchs}
Consider the additive model (\ref{eq:bchs}). Then
\[
 B = \begin{pmatrix}
  \kappa_+ & 0 \\
  (\kappa_+ - \kappa_-) \E[Z_i] & 0 \end{pmatrix}.
\]
If $\kappa_+ = \kappa_-$, then 
\[
 B = \begin{pmatrix}
  \kappa_+ & 0 \\
  0 & 0 \end{pmatrix},
\]
as in \citetalias{battagliaInferenceRegressionVariables2025}.
\end{continuance}

\subsection{Wild Bootstrap with Binary Labels}

Both the fixed-label and coupled-label bootstraps are versions of the following wild bootstrap with binary labels, but differ in how $(\theta^*_i,\hat{\theta}^*_i)$ are drawn:
\begin{equation}\label{eq:boot.wild}
 \begin{aligned}
  Y^*_i & =\hat{\beta}^\prime X^*_i+u^*_i, \\
  X^*_i & =g(\theta^*_i,Z_i), \\
  \hat X^*_i & =g(\hat \theta^*_i,Z_i), \\
  u^*_i & =\hat{u}_i\eta_i,
  \end{aligned}
\end{equation}
where $\eta_i$ are i.i.d.~$(0,1)$ and generated independently of the bootstrap labels $(\theta^*_i,\hat{\theta}^*_i)$, conditional on the original sample. To understand when these methods are valid, we use the following set of conditions on $(\theta_i^*, \hat \theta_i^*)$.

\setcounter{assumption}{0}
\renewcommand{\theassumption}{L}
\begin{assumption}[]\label{ass:L}
\begin{enumerate}[label={(\roman*)}]
   \item\label{ass:L0}\hspace{-5pt}
   $\frac{1}{n}\sum_{i=1}^{n}\E^*[\I[\theta^*_i\ne \hat{\theta}_i]]\to_p 0$.
   \item\label{ass:L1}
  $\frac{1}{\sqrt{n}}\sum_{i=1}^{n}\E^*[\hat{\theta}^*_i(1-\theta^*_i)]\to_p \kappa_+$, $\frac{1}{\sqrt{n}}\sum_{i=1}^{n}\E^*[\theta^*_i(1-\hat{\theta}^*_i)]\to_p \kappa_{-}$.
  \item\label{ass:L2}
  $\V^*\big(\frac{1}{\sqrt{n}}\sum_{i=1}^{n}\hat{\theta}^*_i(1-\theta^*_i)\big)\to_p 0$ and  $\V^*\big(\frac{1}{\sqrt{n}}\sum_{i=1}^{n}\theta^*_i(1-\hat{\theta}^*_i)\big)\to_p 0$.
  \item\label{ass:L3} 
  $\frac{1}{\sqrt{n}}\sum_{i=1}^{n}\hat{\theta}^*_i(1-\theta^*_i)\big(D_{+i}-\E[D_{+i}]\big)+\theta^*_i(1-\hat{\theta}^*_i)\big(D_{-i}-\E[D_{-i}]\big)\to_{p^*} 0$.
\end{enumerate}
\end{assumption}

Assumption~\ref{ass:L}\ref{ass:L0} requires the bootstrap ``true'' labels $\theta^*_i$ to match the estimated labels $\hat{\theta}_i$ on average with probability converging to one. This condition is automatically satisfied if $\theta^*_i=\hat{\theta}_i$.
Assumption~\ref{ass:L}\ref{ass:L1} is the bootstrap analogue of Assumption~\ref{a2}\ref{a2.4}. It requires the average bootstrap false positive and negative rates, when scaled by $\sqrt n$, to converge in probability to $\kappa_{+}$ and $\kappa_{-}$. Together with Assumptions~\ref{ass:L}\ref{ass:L2} and~\ref{ass:L}\ref{ass:L3}, this condition ensures that the bootstrap reproduces the asymptotic bias in the OLS estimator. Assumption~\ref{ass:L}\ref{ass:L3} can be viewed as the bootstrap analogue of Assumption~\ref{a2}\ref{a2.5}. 

The following lemma shows that the high-level Assumptions~\ref{b1} and~\ref{b2} hold for the wild bootstrap with binary labels provided the method for generating $(\theta_i^*, \hat \theta_i^*)$ satisfies Assumption~\ref{ass:L}.

\begin{lemma}\label{lemma_boot_labels} Suppose Assumption~\ref{a2} holds. Let $u^*_i=\hat{u}_i\eta_i$ where $\eta_i$ are i.i.d.~$(0,1)$ independently of $(\theta^*_i,\hat{\theta}^*_i)$ such that $\E^*[|\eta_i|^{2+\delta}]<\infty$ for $\delta>0$. If, in addition, $(\theta^*_i,\hat{\theta}^*_i)$ satisfy Assumption~\ref{ass:L}, then Assumptions~\ref{b1} and~\ref{b2} hold.
\end{lemma}

\subsection{Specific Bootstrap Methods}

Here we study two methods for generating $(\theta^*_i,\hat{\theta}^*_i)$. The first is a standard fixed-regressor bootstrap which sets $\theta^*_i=\hat{\theta}_i$. Its validity requires $\E[D_{+i}|{\theta}_i=0]=\E[D_{+i}]$ and $\E[D_{-i}|{\theta}_i=1]=\E[D_{-i}]$, which holds if $D_{+i}$ and $D_{-i}$ are mean independent of ${\theta}_i$, but not necessarily otherwise. We therefore introduce a coupled-label bootstrap that jointly resamples $(\theta^*_i,\hat{\theta}^*_i)$ independently of the original sample to mimic the bias $B$ without this strong independence assumption.

\subsubsection{Fixed-Label Bootstrap}

This method sets $\theta^*_i=\hat{\theta}_i$ in the bootstrap DGP for $Y^*_i$, treating the bootstrap latent labels as fixed. The bootstrap estimated labels $\hat{\theta}^*_i$ are generated as i.i.d.~Bernoulli random variables, conditional on $\hat{\theta}_i$:
\begin{flalign*}
\p^*(\hat{\theta}^*_i=\theta|\hat{\theta}_i=1)=
\begin{cases}
\hat{F}_{-}/\hat{\pi}& \text{if } \theta=0,\\
1-\hat{F}_{-}/\hat{\pi}& \text{if } \theta=1,
\end{cases}
\end{flalign*}
and
\begin{flalign*}
\p^*(\hat{\theta}^*_i=\theta|\hat{\theta}_i=0)=
\begin{cases}
\hat{F}_{+}/(1-\hat{\pi})& \text{if } \theta=1,\\
1-\hat{F}_{+}/(1-\hat{\pi})& \text{if } \theta=0.
\end{cases}
\end{flalign*}Here, $\hat F_+$ and $\hat F_-$ are such that $\sqrt{n}\hat{F}_+\to_p \kappa_+$ and $\sqrt{n}\hat{F}_-\to_p \kappa_-$, respectively, and $\hat{\pi}=n^{-1}\sum_{i=1}^{n}\hat{\theta}_i$. To this end, we follow \citetalias{battagliaInferenceRegressionVariables2025} and rely on an external sample $(\theta_i, \hat \theta_i)_{i=1}^m$ used to assess the accuracy of the classifier. \citetalias{battagliaInferenceRegressionVariables2025} show that 
\[
 \hat \kappa_+ = \sqrt n \hat F_+, \quad \text{with} \quad \hat F_+ = \frac 1m \sum_{i=1}^m \hat \theta_i(1-\theta_i)
\]
is consistent as $n, m \to \infty$ with $n/m^2 \to 0$ under Assumption~\ref{a2}\ref{a2.4}. An analogous result holds for
\[
 \hat \kappa_- = \sqrt n \hat F_-, \quad \text{with} \quad \hat F_- = \frac 1m \sum_{i=1}^m \theta_i(1-\hat \theta_i).
\]
Importantly, these results allow for the external sample size $m$ to be much smaller than the original sample size $n$ used in the regression.

For this method, Assumption~\ref{ass:L}\ref{ass:L0} is automatically satisfied since $\theta^*_i=\hat{\theta}_i$. One can also show that Assumptions~\ref{ass:L}\ref{ass:L1}-\ref{ass:L2} are satisfied provided $\sqrt{n}\hat{F}_+\to_p \kappa_+$ and $\sqrt{n}\hat{F}_-\to_p \kappa_-$. However, Assumption~\ref{ass:L}\ref{ass:L3} fails whenever $(D_{+i},D_{-i})$ are not mean independent of $\hat{\theta}_i$. This implies that the fixed-label bootstrap only replicates $B$ if we assume that this strong independence condition holds. 
To see this, note that the bootstrap bias is determined by (the probability limit of) the bootstrap expectation of
\[
\hat{B}^*\equiv \frac{1}{\sqrt{n}}\sum_{i=1}^{n}(\hat{\theta}^*_i(1-\theta^*_i)D_{+i}+ \theta^*_i(1-\hat{\theta}^*_i)D_{-i}), ~\text{where~}\theta^*_i=\hat{\theta}_i.
\] Fixing the ``latent'' bootstrap labels $\theta_i^*$ at $\hat{\theta}_i$ implies that the bootstrap expectation of $\hat{B}^*$ is equal to
\begin{flalign*}
\hat{B}\equiv \E^*[\hat{B}^*]&=\hat{\kappa}_+ \frac{1}{1-\hat{\pi}}\frac{1}{n}\sum_{i=1}^{n}(1-\hat{\theta}_i)D_{+i} + \hat{\kappa}_- \frac{1}{\hat{\pi}}\frac{1}{n}\sum_{i=1}^{n}\hat{\theta}_iD_{-i}\\ 
&\to_p \kappa_+ \E[D_{+i}|{\theta}_i=0]+\kappa_{-}\E[D_{-i}|{\theta}_i=1],
\end{flalign*}
whereas
\[
  B = \kappa_+ \E[D_{+i}]  + \kappa_- \E[D_{-i}] .
\]
Hence, $\hat B$ does not consistently estimate the bias $B$ in general, unless $(D_{+i},D_{-i})$ are mean independent of ${\theta}_i$. 
We summarize this result next.
\begin{corollary}\label{CorollaryFixedBoot}
  Under Assumption~\ref{a2}, the fixed-label bootstrap satisfies the conditions of Theorem~\ref{Theorem:generalbootstraptheory} if $\sqrt{n}\hat{F}_ +\to_p \kappa_+$ and $\sqrt{n}\hat{F}_ -\to_p \kappa_-$, and $(D_{+i},D_{-i})$ are mean independent of ${\theta}_i$.
\end{corollary}

\begin{continuance}{ex:bchs}
One instance in which this bootstrap is valid is the special case of model~(\ref{ex:bchs}) where $Z_i$ contains just a constant for the intercept. Then, $D_{+i}$ and $D_{-i}$ are non-stochastic, and so $\hat B \to_p B$.
\end{continuance}

\subsubsection{Coupled-Label Bootstrap}\label{sec:coupled}

The coupled-label bootstrap generates the pair $(\theta_i^*, \hat{\theta}_i^*)$ jointly for each 
observation $i = 1, \ldots, n$, independently across $i$ and independently of 
$u_i^* = \hat{u}_i\eta_i$, according to the following categorical distributions:
\begin{equation*}
\p^*(\theta_i^* = \theta_1, \hat{\theta}_i^* = \theta_2 \mid \hat{\theta}_i = 1) = 
\begin{cases} 
1 - \hat{F}_+ - \frac{\hat{F}_-}{\hat{\pi}} & \text{if } \theta_1 = 1, \theta_2 = 1, \\ 
\hat{F}_- & \text{if } \theta_1 = 1, \theta_2 = 0, \\ 
\hat{F}_+ & \text{if } \theta_1 = 0, \theta_2 = 1, \\ 
\hat{F}_{-}\frac{1-\hat{\pi}}{\hat{\pi}} & \text{if } \theta_1 = 0, \theta_2 = 0,
\end{cases}
\end{equation*}
and%
\begin{equation*}
\p^*(\theta_i^* = \theta_1, \hat{\theta}_i^* = \theta_2 \mid \hat{\theta}_i = 0) = 
\begin{cases} 
\hat{F}_+ \frac{\hat{\pi}}{1-\hat{\pi}} & \text{if } \theta_1 = 1, \theta_2 = 1, \\ 
\hat{F}_- & \text{if } \theta_1 = 1, \theta_2 = 0, \\ 
\hat{F}_+ & \text{if } \theta_1 = 0, \theta_2 = 1, \\ 
1 - \hat{F}_+\frac{1}{1-\hat{\pi}} - \hat{F}_- & \text{if } \theta_1 = 0, \theta_2 = 0.
\end{cases}
\end{equation*}
The key feature of this method is that the probability of discordant pairs $(\theta_i^* , \hat{\theta}_i^*) = (1, 0)$ and  $(0, 1)$ equals $\hat{F}_-$  and $\hat{F}_+$, respectively, for every observation $i$, regardless of $\hat{\theta}_i$. As a result, the bootstrap false positive and negative rates $\E^*[\hat{\theta}_i^*(1-\theta_i^*)] = \hat{F}_+$ and $\E^*[\theta_i^*(1-\hat{\theta}_i^*)] = \hat{F}_-$ are constant across~$i$.  This ensures that the bootstrap correctly replicates the asymptotic bias $B$ without requiring independence between $\theta_i$ and $Z_i$, since the bootstrap misclassification indicators are decoupled from $(D_{+i},D_{-i})$.

The conditional probabilities assigned to the concordant pairs $(\theta_i^*, \hat{\theta}_i^*) = (1,1)$ and $(0,0)$ ensure that the bootstrap Hessian $\hat{Q}^*$ converges to the empirical Hessian $\hat{Q}$ conditional on the data. Since $\hat{Q}^*$ depends only on $\hat{\theta}_i^*$ interacted with $g(1,Z_i)g(1,Z_i)'$ and $g(0,Z_i)g(0,Z_i)'$, its bootstrap expectation is determined by $\E^*[\hat{\theta}_i^*]$. As the discordant pair probabilities are set equal to $\hat{F}_+$ and $\hat{F}_-$, the key requirement is that $\p^*(\theta_i^*=1, \hat{\theta}_i^*=1|\hat{\theta}_i=1)\to_p 1$ and $\p^*(\theta_i^*=1, \hat{\theta}_i^*=1|\hat{\theta}_i=0)\to_p 0$, both of which are satisfied by the proposed construction. Under this distribution, $\E^*[\hat{\theta}_i^*]=\hat{\theta}_i(1-\hat{F}_-/\hat{\pi})+(1-\hat{\theta}_i)\hat{F}_+/(1-\hat{\pi})$, which ensures that the coupled-label bootstrap Hessian matches the fixed-label bootstrap Hessian in expectation. 
Other choices are nevertheless possible.

The following result formally establishes the validity of this bootstrap:

\begin{theorem}\label{TheoremCoupledB}
Suppose that Assumption~\ref{a2} holds. Let $u^*_i=\hat{u}_i\eta_i$ where $\eta_i$ are i.i.d.~$(0,1)$ independently of $(\theta^*_i,\hat{\theta}^*_i)$ such that $\E^*[|\eta_i|^{2+\delta}]<\infty$ for $\delta>0$. If, in addition, $\sqrt{n}\hat{F}_ +\to_p \kappa_+$ and $\sqrt{n}\hat{F}_ -\to_p \kappa_-$, then the coupled-label bootstrap satisfies the conditions of Theorem~\ref{Theorem:generalbootstraptheory} and, consequently, 
\begin{flalign*}
  \sqrt{n}\big(\hat{\beta}^*-\hat{\beta}\big)\rightarrow_{d^*} N(b,V),
\end{flalign*}
as $n \to \infty$, with $b = - Q^{-1} B \beta$ and $V = Q^{-1} \Sigma Q^{-1}$.
\end{theorem}
Theorem~\ref{TheoremCoupledB} shows that the coupled-label bootstrap correctly replicates the bias and variance of the asymptotic distribution of the OLS estimator in Theorem~\ref{t1}, thus justifying the use of percentile intervals for valid inference. However, as our simulations show, finite-sample coverage can deviate from the nominal level. We therefore recommend using the coupled-label bootstrap with the two finite-sample modifications described in the next section.

\section{Finite-Sample Modifications}\label{sec:finite-sample}

The asymptotic validity of the bootstrap percentile interval requires the bootstrap to correctly mimic the asymptotic distribution of $\sqrt{n}(\hat{\beta}-\beta)$, including its mean and variance. From the proofs of Lemma~\ref{lemma_B1(ii)} and Theorem~\ref{Theorem:generalbootstraptheory},
\begin{equation}\label{eq:expansion.bootstrap}
 \sqrt n(\hat \beta^* - \hat \beta ) = Z^* + \hat{b}^*,
\end{equation}
where $Z^* \to_{d^*} N(0,V)$ and
\begin{flalign*}
 \hat{b}^* &= -\hat{Q}^{*-1}\bigg(\frac{1}{\sqrt{n}}\sum_{i=1}^{n}\hat{\theta}^*_i(1-\theta^*_i)D_{+i} + \theta^*_i(1-\hat{\theta}^*_i)D_{-i}\bigg)\hat{\beta}\equiv -\hat{Q}^{*-1}\hat{B}^*\hat{\beta},
\end{flalign*} where $\hat{Q}^*=n{-1}\sum_{i=1}^{n}\hat{X}^*_i\hat{X}^{*\prime}_i$ is the bootstrap Hessian matrix.
To show that $\hat{b}^*\to_{p^*}b$, we add and subtract appropriately to obtain
\begin{flalign}\label{eq:expansion.bootstrap2}
 \hat{b}^* &= -\frac{1}{\sqrt{n}}\sum_{i=1}^{n}\hat{\theta}^*_i(1-\theta^*_i)\Gamma_{+}\beta + \theta^*_i(1-\hat{\theta}^*_i)\Gamma_{-}\beta + o_{p^*}(1)\equiv b^*+o_{p^*}(1),
\end{flalign}where $\Gamma_{+}=Q^{-1}\E[D_{+i}]$ and $\Gamma_{-}=Q^{-1}\E[D_{-i}]$.
While the coupled-label bootstrap bias $b^*$ is consistent for $b$, finite-sample distortions may arise. In particular, as noted in \citetalias{battagliaInferenceRegressionVariables2025}, estimating $F_+$ and $F_-$ introduces additional noise that can inflate the variance of $b^*$, requiring a variance correction to improve coverage of bias-corrected confidence intervals when $m/n$ is small. The same issue arises here, but a similar correction can easily be incorporated into the coupled-label bootstrap. A second source of distortion arises from the $o_{p^*}(1)$ term representing the difference between $\hat{b}^*$ and $b^*$. This term depends on the difference between $\hat{Q}^{*-1}$ and $\hat{Q}^{-1}$, which can be non-negligible in finite samples, especially when the proportion of $\hat{\theta}_i$ equal to 1 or 0 is small. In such cases, the bootstrap mean of $\hat{b}^*$ is different from that of $b^*$, introducing a location bias in the bootstrap. This motivates rotating $\hat{\beta}^*-\hat{\beta}$ by $\hat{R}^*=\hat{Q}^{-1}\hat{Q}^*$. 

We discuss these two modifications next. We first describe the variance correction for the coupled-label bootstrap, establish its validity, and show it is able to replicate the two leading higher-order variance terms. We then describe the rotation that resolves the mismatch between the bootstrap Hessian and its sample analogue.

\subsection{Coupled-Label Bootstrap with Variance Correction}

To generate each bootstrap sample, first draw $V_+^* \sim \text{Binomial}(m, \hat F_+)$ and $V_-^* \sim \text{Binomial}(m, \hat F_-)$ and set $\hat F_+^* = V_+^*/m$ and $\hat F_-^* = V_-^*/m$. 
Conditional on $\hat F_+^*$ and $\hat F^*_-$, the pairs $(\theta_i^*, \hat{\theta}_i^*)$ are generated jointly for each observation $i = 1, \ldots, n$, independently across $i$ and independently of 
$u_i^* = \hat{u}_i\eta_i$, according to the following distributions:
\begin{equation}\label{eq:coupled.1.var}
\p^*(\theta_i^* = \theta_1, \hat{\theta}_i^* = \theta_2 \mid \hat{\theta}_i = 1, \hat F^*_+, \hat F^*_-) = 
\begin{cases} 
1 - \hat{F}_+^* - \frac{\hat{F}_-^*}{\hat{\pi}} & \text{if } \theta_1 = 1, \theta_2 = 1, \\ 
\hat{F}_-^* & \text{if } \theta_1 = 1, \theta_2 = 0, \\ 
\hat{F}_+^* & \text{if } \theta_1 = 0, \theta_2 = 1, \\ 
\hat{F}_{-}^*\frac{1-\hat{\pi}}{\hat{\pi}} & \text{if } \theta_1 = 0, \theta_2 = 0,
\end{cases}
\end{equation}
and%
\begin{equation}\label{eq:coupled.2.var}
\p^*(\theta_i^* = \theta_1, \hat{\theta}_i^* = \theta_2 \mid \hat{\theta}_i = 0, \hat F^*_+, \hat F^*_-) = 
\begin{cases} 
\hat{F}_+^* \frac{\hat{\pi}}{1-\hat{\pi}} & \text{if } \theta_1 = 1, \theta_2 = 1, \\ 
\hat{F}_-^* & \text{if } \theta_1 = 1, \theta_2 = 0, \\ 
\hat{F}_+^* & \text{if } \theta_1 = 0, \theta_2 = 1, \\ 
1 -  \hat{F}_+^*\frac{1}{1-\hat{\pi}} - \hat{F}_-^* & \text{if } \theta_1 = 0, \theta_2 = 0.
\end{cases}
\end{equation}
Note that the $\hat F^*_+$ and $\hat F_-^*$ are generated only once per bootstrap sample and are held fixed when generating the $n$ pairs $(\theta_i^*, \hat{\theta}_i^*)$ for $i = 1,\ldots,n$.

The following result formally establishes validity of this bootstrap:

\begin{theorem}\label{TheoremCoupledB.variance}
Suppose that Assumption~\ref{a2} holds. Let $u^*_i=\hat{u}_i\eta_i$ where $\eta_i$ are i.i.d.~$(0,1)$ independently of $(\theta^*_i,\hat{\theta}^*_i)$ such that $\E^*[|\eta_i|^{2+\delta}]<\infty$ for $\delta>0$. If, in addition, $\sqrt{n}\hat{F}_ +\to_p \kappa_+$ and $\sqrt{n}\hat{F}_ -\to_p \kappa_-$ and $n/m^2 \to 0$, then the coupled-label bootstrap with variance correction satisfies the conditions of Theorem~\ref{Theorem:generalbootstraptheory} and, consequently, 
\begin{flalign*}
  \sqrt{n}\big(\hat{\beta}^*-\hat{\beta}\big)\rightarrow_{d^*} N(b,V),
\end{flalign*}
as $n \to \infty$, with $b = - Q^{-1} B \beta$ and $V = Q^{-1} \Sigma Q^{-1}$.
\end{theorem}

Theorem~\ref{TheoremCoupledB.variance} shows that this bootstrap is able to reproduce the asymptotic bias and variance of the OLS estimator. We now show it is able to match key higher-order variance terms as well.

\subsection{Variance Analysis}

\subsubsection{Finite-Sample Correction from \citetalias{battagliaInferenceRegressionVariables2025}}

We first consider the bias-corrected estimator proposed in \citetalias{battagliaInferenceRegressionVariables2025}. That estimator can be written as
\[
 \hat \beta^{bc} = (I + \hat \Gamma_+ \hat F_+ + \hat \Gamma_- \hat F_-) \hat \beta,
\]
with $\hat Q = \frac 1n \sum_{i=1}^n \hat X_i \hat X_i'$, $\hat \Gamma_+ = \hat Q^{-1} \left( \frac 1n \sum_{i=1}^n D_{+i} \right)$, and $\hat \Gamma_-$ defined analogously. We may analyze the effect of $\hat F_+$ and $\hat F_-$ on the variance of $\hat \beta^{bc}$ as follows. First, as $\V(\sqrt n(\hat \beta - \beta)) \approx V$, we have
\[
 \begin{aligned}
 \E[\V(\hat \beta^{bc} | \hat F_+, \hat F_-)] & \approx  \frac 1n \E[ (I + \Gamma_+ \hat F_+ + \Gamma_- \hat F_-) V (I + \Gamma_+ \hat F_+ + \Gamma_- \hat F_-)'] \\[4pt]
 & = \frac 1n \left( V + O(n^{-1/2}) \right) ,
 \end{aligned}
\]
because $F_+ = \E[\hat F_+]$, $F_- = \E[\hat F_-]$ and $\sqrt n F_+, \sqrt n F_- = O(1)$.
Moreover, 
\[
 \begin{aligned}
 \V(\E[\hat \beta^{bc} |\hat F_+,\hat F_-]) & \approx \V((I + \Gamma_+ \hat F_+ + \Gamma_- \hat F_-) \beta) \\[4pt]
& = \frac{ F_+(1-F_+)}{m} \Gamma_+ \beta \beta' \Gamma_+' + \frac{ F_-(1-F_-)}{m} \Gamma_- \beta \beta' \Gamma_-' \\
& \quad   - \frac{ F_+ F_- }{m}  ( \Gamma_+ \beta \beta' \Gamma_-' + \Gamma_- \beta \beta' \Gamma_+'),
 \end{aligned}
\]
with $F_+ = \E[\hat F_+]$ and $F_- = \E[\hat F_-]$.
The law of total variance implies that the variance of $\hat \beta^{bc}$ is approximately equal to the sum of these two terms. As $\sqrt n F_+, \sqrt n F_- = O(1)$ and $\sqrt{n}/m = o(1)$, we have
\begin{multline}\label{eq:var.bc.higher.order}
 \V(\sqrt n(\hat \beta^{bc} - \beta)) \\
 \approx V + \frac{n F_+(1-F_+)}{m} \Gamma_+ \beta \beta' \Gamma_+' + \frac{n F_-(1-F_-)}{m} \Gamma_- \beta \beta' \Gamma_-' + O(n^{-1/2}) .
\end{multline}
 Under the framework adopted in \citetalias{battagliaInferenceRegressionVariables2025} where $m/n \to 0$ with $n/m^2 \to 0$, the second and third terms are $O(\sqrt n/m)$ and therefore dominate the remaining $O(n^{-1/2})$ term. Thus, it is these two terms that are most critical to match for the purposes of finite-sample variance correction. 

\subsubsection{The Bootstrap Matches the Finite-Sample Variance Correction}

We now show that the coupled-label bootstrap with variance correction is able to match the above two $O(\sqrt n/m)$ terms. Recalling the bootstrap expansion (\ref{eq:expansion.bootstrap2}), the bootstrap expectation of the bias term $b^*$ is 
\[
 \E^*[b^*] = - \left( \sqrt n \hat F_+ \Gamma_+ \beta + \sqrt n \hat F_- \Gamma_- \beta \right) \to_p b.
\]
Therefore, drawing $\hat F^*_+$ and $\hat F^*_-$ from a distribution centered at $\hat F_+$ and $\hat F_-$ does not affect the bootstrap's ability to reproduce the leading bias term.
 
To analyze the bootstrap variance of the bias term, first note that conditional on the data and $\hat F_+^*, \hat F_-^*$, the $\hat{\theta}^*_i(1-\theta^*_i)$ and $\theta^*_i(1-\hat{\theta}^*_i)$ are i.i.d. Bernoulli and satisfy
\[
 \begin{aligned}
 \E^*[\hat{\theta}^*_i(1-\theta^*_i)|\hat F_+^*, \hat F_-^*] & = \hat F_+^* , & & & 
 \E^*[\theta^*_i(1-\hat{\theta}^*_i)|\hat F_+^*, \hat F_-^*] & = \hat F_-^* , \\
 \V^*(\hat{\theta}^*_i(1-\theta^*_i)|\hat F_+^*, \hat F_-^*) & = \hat F_+^*(1-\hat F_+^*), & & & 
 \V^*(\theta^*_i(1-\hat{\theta}^*_i)|\hat F_+^*, \hat F_-^*) & = \hat F_-^*(1-\hat F_-^*),
 \end{aligned}
\]
and 
\[
 \C^*(\hat{\theta}^*_i(1-\theta^*_i),\theta^*_i(1-\hat{\theta}^*_i)|\hat F_+^*, \hat F_-^*) = - \hat F_+^* \hat F_-^*.
\]
Using these expressions and the fact that $m\hat F^*_+$ and $m\hat F^*_-$ are independent Binomial$(m,\hat F_+)$ and Binomial$(m,\hat F_-)$ random variables conditional on the data, we have
\[
 \begin{aligned}
 \V^*\left( \E^*[b^*|\hat F_+^*,\hat F^*_-] \right)
 & = \V^*\left( \sqrt n \hat F^*_+ \Gamma_+ \beta + \sqrt n \hat F^*_- \Gamma_- \beta \right) \\
 & = \frac{n \hat F_+(1-\hat F_+)}{m}\Gamma_+ \beta \beta' \Gamma_+' + \frac{n \hat F_-(1-\hat F_-)}{m}\Gamma_- \beta \beta' \Gamma_-'.
 \end{aligned}
\]
Moreover,
\begin{multline*}
 \V^*(b^*|\hat F_+^*,\hat F^*_-)
 = \hat F_+^*(1-\hat F_+^*) \Gamma_+ \beta \beta' \Gamma_+' + \hat F_-^*(1-\hat F_-^*) \Gamma_- \beta \beta' \Gamma_-' \\ 
 - \hat F_+^* \hat F_-^* ( \Gamma_+ \beta \beta' \Gamma_-' + \Gamma_- \beta \beta' \Gamma_+'),
\end{multline*}
from which we see that 
\[
\begin{aligned}
 \E^* \left[ \V^*(b^*|\hat F_+^*,\hat F^*_-) \right]
 & = \left( 1 - \frac{1}{m} \right) \left( \hat F_+(1-\hat F_+) \Gamma_+ \beta \beta' \Gamma_+' + \hat F_-(1-\hat F_-) \Gamma_- \beta \beta' \Gamma_-' \right) \\
 & \quad - \hat F_+ \hat F_- ( \Gamma_+ \beta \beta' \Gamma_-' + \Gamma_- \beta \beta' \Gamma_+') \\
 & = O_p(n^{-1/2}).
\end{aligned}
\]
Hence, by the law of total variance, we have,
\[
 \V^*(b^*) = \frac{n \hat F_+(1-\hat F_+)}{m}\Gamma_+ \beta \beta' \Gamma_+' + \frac{n \hat F_-(1-\hat F_-)}{m}\Gamma_- \beta \beta' \Gamma_-' + O_p(n^{-1/2}).
\]
And since $\sqrt{n}\hat{F}_ +\to_p \kappa_+$ and $\sqrt{n}\hat{F}_ -\to_p \kappa_-$, we have 
\[
 \V^*(b^*) = \frac{n F_+(1-F_+)}{m}\Gamma_+ \beta \beta' \Gamma_+' + \frac{n F_-(1-F_-)}{m}\Gamma_- \beta \beta' \Gamma_-' + o_p(\sqrt n/m) .
\]
These leading terms match $O(\sqrt n/m)$ terms in~\eqref{eq:var.bc.higher.order}, as required.

\subsection{Rotation}
As discussed previously, an additional source of finite-sample distortions associated with the coupled-label bootstrap (and its variance-corrected variant) is the discrepancy between the inverse bootstrap Hessian matrix $\hat{Q}^{*-1}$ and its sample analogue $\hat{Q}^{-1}$. If this difference is large, the remainder term in the bootstrap expansion (\ref{eq:expansion.bootstrap2}) is not negligible, introducing a bias in the bootstrap distribution. 

One way to solve this problem is to bootstrap the distribution of 
$\sqrt{n}\hat{R}^*(\hat{\beta}^*-\hat{\beta})$, where $\hat{R}^*=\hat{Q}^{-1}\hat{Q}^*$. This is equivalent to bootstrapping the distribution of 
\[
\sqrt{n}(\tilde{\beta}^*-\hat{R}^*\hat{\beta}),\quad \tilde{\beta}^*=(\hat{X}^{\prime}\hat{X})^{-1}\hat{X}^{*\prime}Y^*,
\]where $\hat{X}^*$ and $Y^*$ are obtained by the variance-corrected coupled-label bootstrap.
With this rotation, the bias term of the bootstrap expansion in (\ref{eq:expansion.bootstrap}) becomes
\begin{flalign*}
 \hat{R}^*\hat{b}^* &= -\hat{Q}^{-1}\frac{1}{\sqrt{n}}\sum_{i=1}^{n}\hat{\theta}^*_i(1-\theta^*_i)D_{+i}\hat{\beta} + \theta^*_i(1-\hat{\theta}^*_i)D_{-i}\hat{\beta},
\end{flalign*}thus no longer depending on $\hat{Q}^{*-1}$. 
Since $\hat{R}^*\to_{p^*} I_k$, the rotated coupled-label bootstrap with variance correction is asymptotically valid under the same conditions as those stated in Theorem~\ref{TheoremCoupledB.variance}.

\section{Simulations}\label{sec:simulations}
We analyze the finite-sample properties of the bootstrap methods by
extending the simulation design of %
\citetalias{battagliaInferenceRegressionVariables2025} in two directions.
First, we consider an interactions model in which a standard normal
covariate $Z_{i}$ is interacted with the binary label $\theta _{i}$ so that
the effect of $Z_{i}$ on $Y_{i}$ differs between $\theta _{i}=1$ and $\theta
_{i}=0$. Second, we allow $p_{i}=\mathbb{P}(\theta _{i}=1)$ to vary across
observations to introduce correlation between $\theta _{i}$ and the matrices 
$D_{+i}$ and $D_{-i}$, which, as shown above, invalidates the fixed-label
bootstrap. The model is: 
\begin{align*}
Y_{i}& =10+\theta _{i}Z_{i}+Z_{i}+(0.3+0.2\,\theta _{i})\,u_{i}, \\
Z_{i}& \sim N(0,1),\qquad u_{i}\sim \text{i.i.d.}\ N(0,1), \\
p_{i}& =F(Z_{i}^{2}),
\end{align*}%
where $F\left( \cdot \right) $ is the cumulative distribution function of a $%
\chi ^{2}\left( 1\right) $ random variable, which implies that $%
p_{i}\thicksim U[0,1].$ The joint distribution of the label pair $(\theta
_{i},\hat{\theta}_{i})^{\prime }$ is given by 
\begin{align*}
\mathbb{P}(\theta _{i}=1,\,\hat{\theta}_{i}=1)& =\tilde{p}_{i}-F_{-}, & 
\mathbb{P}(\theta _{i}=1,\,\hat{\theta}_{i}=0)& =F_{-}, \\
\mathbb{P}(\theta _{i}=0,\,\hat{\theta}_{i}=1)& =F_{+}, & \mathbb{P}(\theta
_{i}=0,\,\hat{\theta}_{i}=0)& =1-\tilde{p}_{i}-F_{+}.
\end{align*}%
We set $\kappa _{+}=\kappa _{-}=\kappa $ and $F_{+}=F_{-}=\kappa /\sqrt{n}$.
To ensure all probabilities remain non-negative, we map $p_{i}$ into 
\begin{equation*}
\tilde{p}_{i}=p_{i}\cdot 2(\bar{p}-F_{+})+F_{+},
\end{equation*}%
so that $\tilde{p}_{i}$ is uniformly distributed between $F_{+},$ and $2\bar{%
p}-F_{+}$, where $\bar{p}$ is the average probability that $\theta _{i}=1$.

We consider two values of $\bar{p}\in \{0.05,\,0.5\}$, matching the constant
probabilities used by \citetalias{battagliaInferenceRegressionVariables2025}%
. The correlation between $\theta _{i}$ and $Z_{i}^{2}$ depends strongly on $%
\bar{p}$: it is approximately $0.02$ when $\bar{p}=0.05$ but rises to
approximately $0.31$ when $\bar{p}=0.5$. As theory predicts, this has
important consequences for the performance of the fixed-label bootstrap,
which requires mean independence between $\theta _{i}$ and $(D_{+i}$, $%
D_{-i})$.

The bootstrap DGP takes the form 
\begin{flalign*}
Y_i^* &= \hat\beta'X^*_i + u_i^*, \qquad u_i^* = \hat{u}_i \eta_i, \quad \eta_i \sim \text{i.i.d.}\ N(0,1),\\
X^*_i&=(1,\theta_i^* Z_i, Z_i)^\prime,
\end{flalign*}where $\hat{\beta}$ is the sample OLS estimator and $\hat{u}%
_{i}$ are the OLS residuals. In all experiments, we use the wild bootstrap
with standard normal weights to generate $u_{i}^{\ast }$. Four methods for
generating $(\theta _{i}^{\ast },\hat{\theta}_{i}^{\ast })$ are considered.
One is a naive bootstrap that sets $\theta _{i}^{\ast }=\hat{\theta}%
_{i}^{\ast }=\hat{\theta}_{i}$. This bootstrap method is denoted as
\textquotedblleft no-label resampling\textquotedblright\ in the tables. By
not introducing any measurement error in the bootstrap world, it implicitly
sets the bootstrap bias to zero. This method is only included as a benchmark
to demonstrate the need to resample labels in the bootstrap world. We then
consider the fixed-label and coupled-label methods from Section~\ref%
{sec:labels}. Finally, we also include the variant of the coupled-label
bootstrap with the finite-sample modifications discussed in Section~\ref%
{sec:finite-sample}. We use $B=499$ bootstrap replications and $10{,}000$
Monte Carlo replications throughout. The sample sizes we consider are $n\in
\{8{,}000;\,16{,}000;\,32{,}000\}$, and the misclassification rates are $%
\kappa \in \{0.5,\,1,\,1.5\}.$ We vary the size of the external sample $m$ to
keep the ratio $\sqrt{n}/m$ constant at 0.1265, corresponding to $m=707$ for 
$n=8{,}000,$ $m=1{,}000$ for $n=16{,}000$, and $m=1{,}414$ for $n=32{,}000.$

We focus on inference for the slope coefficient on the interaction term $%
\theta_i Z_i$. For each configuration we report three statistics: the median
bias, the empirical coverage rate of a nominal 95\% confidence interval, and
the median interval length. For bootstrap methods, we report equal-tailed
percentile intervals, which dominate symmetric intervals in this setting.

Table~\ref{tab:1} presents results for $\bar{p}=0.5$. Each panel contains
six estimators: OLS using the imputed regressor $\hat\theta_i Z_i$, the
bias-corrected estimator of %
\citetalias{battagliaInferenceRegressionVariables2025} using their finite-sample variance adjustment, and the four
bootstrap methods.

The use of an imputed label biases OLS downward, with the bias increasing in 
$\kappa $. The resulting undercoverage is severe: even in the mildest
configuration ($n=8{,}000$, $\kappa =0.5$), OLS coverage is only $15.0\%$.
The \citetalias{battagliaInferenceRegressionVariables2025} method removes virtually all the bias and delivers reliable inference
across configurations, though coverage deteriorates modestly as $\kappa $
rises, and the intervals are substantially wider than those based on OLS.

The no-label-resampling bootstrap simply mimics the OLS estimator,
inheriting its bias and producing similarly poor coverage (at best %
15.5\%).

The fixed-label bootstrap overcorrects: it introduces bias in the opposite
direction. Coverage remains well below nominal and far below the %
\citetalias{battagliaInferenceRegressionVariables2025} benchmark.

The coupled-label bootstrap, by contrast, correctly replicates the bias of
the bias-corrected estimator. Coverage remains below nominal level, however,
because the basic coupled-label bootstrap does not account for the
additional sampling uncertainty from estimating $F_{+}$ and $F_{-}$:
interval lengths are close to those of OLS, and the improvement in coverage
relative to no-label resampling comes entirely from correctly centering the
bootstrap distribution. The coupled-label bootstrap with rotation and
variance correction achieves coverage closest to the nominal $95\%$ among
all methods, slightly edging the %
\citetalias{battagliaInferenceRegressionVariables2025} analytic correction, while maintaining similar interval lengths.

\begin{table}[t]
\caption{Interactions model, $\bar{p}=0.5$}
\label{tab:1}{\footnotesize \centering
\makebox[\textwidth][c]{%
\begin{threeparttable}
\begin{tabular}{l rrr rrr rrr}
\toprule
& \multicolumn{3}{c}{Median bias}
& \multicolumn{3}{c}{Coverage (\%)}
& \multicolumn{3}{c}{Interval length} \\
\cmidrule(lr){2-4}\cmidrule(lr){5-7}\cmidrule(lr){8-10}
$\kappa_+ = \kappa_- =$
  & \multicolumn{1}{c}{$0.5$} & \multicolumn{1}{c}{$1$} & \multicolumn{1}{c}{$1.5$}
  & \multicolumn{1}{c}{$0.5$} & \multicolumn{1}{c}{$1$} & \multicolumn{1}{c}{$1.5$}
  & \multicolumn{1}{c}{$0.5$} & \multicolumn{1}{c}{$1$} & \multicolumn{1}{c}{$1.5$} \\
\midrule
\multicolumn{10}{l}{\textit{$n = 8{,}000$}} \\[2pt]
OLS				     & $-$0.04 & $-$0.07 & $-$0.11 & 15.5 & 0.0 & 0.0 & 0.05& 0.06 & 0.06 \\
BCHS bias-corrected \& var.\ adj.\                     & $-$0.00 & $-$0.01 & $-$0.01 & 93.8 & 93.8 & 92.3 & 0.08 & 0.10 & 0.12 \\
No-label resampling      	     & $-$0.04 & $-$0.07 & $-$0.11 & 15.5 & 0.0 & 0.0 & 0.05& 0.06 & 0.06 \\
Fixed-label bootstrap                & \phantom{$$}0.01 & \phantom{$$}0.01 & $$0.01 & 73.2 & 69.5 & 70.8 & 0.06 & 0.07 & 0.08 \\
Coupled-label bootstrap  &{$-$}0.00 & $-$0.01 & $-$0.02 & 85.1 & 80.6 & 72.6 & 0.06 & 0.07 & 0.07 \\
Coupled-label, rotation \& var.\ adj.\ & {$-$}0.00 & $-$0.01 & $-$0.01 & 94.5 & 94.6 & 92.7 & 0.08 & 0.11 & 0.12 \\
\midrule
\multicolumn{10}{l}{\textit{$n = 16{,}000$}} \\[2pt]
OLS					     & $-$0.03 & $-$0.05 & $-$0.08 & 11.5 & 0.0 & 0.0 & 0.03& 0.04 & 0.04 \\
BCHS bias-corrected \& var.\ adj.\                     & {$-$}0.00 & $-$0.00 & $-$0.01 & 94.2 & 93.8 & 92.8 & 0.05 & 0.07 & 0.08 \\
No-label resampling		& $-$0.03 & $-$0.05 & $-$0.08 & 11.5 & 0.0 & 0.0 & 0.03& 0.04 & 0.04 \\
Fixed-label bootstrap                & \phantom{$$}0.01& \phantom{$$}0.01& $$0.01 & 69.5 & 62.2 & 61.1 & 0.04 & 0.05 & 0.05 \\
Coupled-label bootstrap              & {$-$}0.00 & $-$0.01 & $-$0.01 & 83.6& 78.4 & 72.4 & 0.04 & 0.04 & 0.05 \\
Coupled-label, rotation \& var.\ adj.\ & {$-$}0.00 & $-$0.00 & $-$0.01 & 94.7 & 94.7 & 93.3 & 0.06 & 0.07 & 0.08 \\
\midrule
\multicolumn{10}{l}{\textit{$n = 32{,}000$}} \\[2pt]
OLS					     & $-$0.02 & $-$0.04 & $-$0.06 & 9.6 & 0.0 & 0.0 & 0.02& 0.02 & 0.03 \\
BCHS bias-corrected \& var.\ adj.\                     & {$-$}0.00 & $-$0.00 & $-$0.00 & 93.7 & 93.8 & 93.1 & 0.04 & 0.05 & 0.06 \\
No-label resampling				     & $-$0.02 & $-$0.04 & $-$0.06 & 9.7 & 0.0 & 0.0 & 0.02& 0.02 & 0.03 \\
Fixed-label bootstrap                & \phantom{$$}0.01 & \phantom{$$}0.01 & \phantom{$$}0.01 & 64.9 & 56.7 & 52.8 & 0.03 & 0.03 & 0.03 \\
Coupled-label bootstrap            & {$-$}0.00 &{$-$}0.00 & $-$0.01 & 80.6 & 75.5 & 71.0 & 0.02 & 0.03 & 0.03 \\
Coupled-label, rotation \& var.\ adj.\ & {$-$}0.00 & $-$0.00 & $-$0.00 & 94.2 & 94.5 & 93.8 & 0.04 & 0.05 & 0.06 \\\bottomrule
\end{tabular}
\begin{tablenotes}[flushleft]
\setlength{\labelsep}{0pt}
\item \textit{Note:} 10{,}000 Monte Carlo replications; $B=499$ bootstrap replications. 
\end{tablenotes}
\end{threeparttable}} }
\end{table}

\begin{table}[t]
\caption{Interactions model, $\bar{p}=0.05$}
\label{tab:2}{\footnotesize \centering
\makebox[\textwidth][c]{%
\begin{threeparttable}
\begin{tabular}{l rrr rrr rrr}
\toprule
& \multicolumn{3}{c}{Median bias}
& \multicolumn{3}{c}{Coverage (\%)}
& \multicolumn{3}{c}{Interval length} \\
\cmidrule(lr){2-4}\cmidrule(lr){5-7}\cmidrule(lr){8-10}
$\kappa_+ = \kappa_- =$
  & \multicolumn{1}{c}{$0.5$} & \multicolumn{1}{c}{$1$} & \multicolumn{1}{c}{$1.5$}
  & \multicolumn{1}{c}{$0.5$} & \multicolumn{1}{c}{$1$} & \multicolumn{1}{c}{$1.5$}
  & \multicolumn{1}{c}{$0.5$} & \multicolumn{1}{c}{$1$} & \multicolumn{1}{c}{$1.5$} \\
\midrule
\multicolumn{10}{l}{\textit{$n = 8{,}000$}} \\[2pt]
OLS					     & $-$0.08 & $-$0.16 & $-$0.25 & 15.0 & 0.0 & 0.0 & 0.10& 0.13 & 0.14 \\
BCHS bias-corrected \& var.\ adj.\                     & $-$0.01 & $-$0.03 & $-$0.07 & 93.6 & 91.4 & 84.6 & 0.18 & 0.25 & 0.31 \\
No-label resampling      	     & $-$0.08 & $-$0.16 & $-$0.25 & 15.4 & 0.0 & 0.0 & 0.10& 0.12 & 0.14 \\
Fixed-label bootstrap                & {$-$}0.00 & {$-$}0.02 & {$-$}0.05 & 83.2 & 78.3 & 66.6 & 0.12 & 0.15 & 0.16 \\
Coupled-label bootstrap  & {$-$}0.00 & $-$0.02 & $-$0.05 & 82.6 & 77.5 & 65.9 & 0.12 & 0.15 & 0.16 \\
Coupled-label, rotation \& var.\ adj.\ & {$-$}0.01 & $-$0.03 & $-$0.07 & 93.2 & 90.4 & 79.4 & 0.17 & 0.22 & 0.25 \\
\midrule
\multicolumn{10}{l}{\textit{$n = 16{,}000$}} \\[2pt]
OLS					     & $-$0.05 & $-$0.11 & $-$0.17 & 12.2 & 0.0 & 0.0 & 0.07& 0.08 & 0.09 \\
BCHS bias-corrected \& var.\ adj.\                     & {$-$}0.00 & $-$0.01 & $-$0.03 & 93.5 & 92.7 & 89.6 & 0.12 & 0.17 & 0.20 \\
No-label resampling		& $-$0.05 & $-$0.11 & $-$0.17 & 12.0 & 0.0 & 0.0 & 0.07& 0.08 & 0.09 \\
Fixed-label bootstrap                & {$-$}0.00& {$-$}0.01& $-$0.02 & 81.5 & 76.9 & 71.7 & 0.08 & 0.10 & 0.11 \\
Coupled-label bootstrap              & {$-$}0.00 & $-$0.01 & $-$0.02 & 80.8& 76.2 & 71.0 & 0.08 & 0.10 & 0.11 \\
Coupled-label, rotation \& var.\ adj. & {$-$}0.00 & $-$0.01 & $-$0.03 & 93.6 & 92.5 & 87.8 & 0.12 & 0.16 & 0.18 \\
\midrule
\multicolumn{10}{l}{\textit{$n = 32{,}000$}} \\[2pt]
OLS					     & $-$0.04 & $-$0.08 & $-$0.12 & 9.4 & 0.0 & 0.0 & 0.05& 0.05 & 0.06 \\
BCHS bias-corrected \& var.\ adj.\                     & {$-$}0.00 & $-$0.01 & $-$0.02 & 92.6 & 92.5 & 90.9 & 0.08 & 0.11 & 0.14 \\
No-label resampling				     & $-$0.04 & $-$0.08 & $-$0.12 & 9.5 & 0.0 & 0.0 & 0.05& 0.05 & 0.06 \\
Fixed-label bootstrap                & {$-$}0.00 & {$-$}0.00 & {$-$}0.01 & 78.7 & 73.5 & 69.5 & 0.05 & 0.06 & 0.07 \\
Coupled-label bootstrap            & {$-$}0.00 & {$-$}0.00& $-$0.01 & 78.0 & 72.6 & 68.6 & 0.05 & 0.06 & 0.07 \\
Coupled-label, rotation \& var.\ adj.\ & {$-$}0.00 & $-$0.01 & $-$0.02 & 93.0 & 92.5 & 90.3 & 0.08 & 0.11 & 0.13 \\
\bottomrule
\end{tabular}
\begin{tablenotes}[flushleft]
\setlength{\labelsep}{0pt}
\item \textit{Note:} See notes to Table~\ref{tab:1}. With $\bar{p}=0.05$, the correlation between $\theta_i$ and $(D_{+i},D_{-i})$ is low ($\approx 0.02$), so the fixed-label bootstrap performs relatively well. 
\end{tablenotes}
\end{threeparttable}} }
\end{table}

Table~\ref{tab:2} presents results for $\bar{p}=0.05$. This near-singular
design is very challenging: the sample contains very few observations with $%
\theta _{i}=1$, and the majority of these are misclassified when $\kappa
\geq 1$. As a consequence, the OLS bias is roughly twice as large as in the $%
\bar{p}=0.5$ case. Because the correlation between $\theta _{i}$ and $%
(D_{+i},D_{-i})$ is low, the fixed-label bootstrap performs well and even slightly
better than the coupled-label bootstrap. However, both maintain coverage
below the \citetalias{battagliaInferenceRegressionVariables2025} benchmark.
The variance correction and Hessian rotation improve coverage, and
the gap with the analytic correction from %
\citetalias{battagliaInferenceRegressionVariables2025} decreases with sample
size.

We conclude that the coupled-label bootstrap with finite-sample
modifications is the preferred bootstrap method in all situations as it delivers
coverage close to nominal coverage while producing confidence intervals of
competitive length.

\section{Application}\label{sec:application}

As an illustration, we revisit \citetalias{battagliaInferenceRegressionVariables2025}' regressions investigating the wage premium for remote work. The data are from \cite{hansenRemoteWorkJobs2026}, who construct a dataset that imputes remote work status for a corpus of online job postings provided by Lightcast.\footnote{See \url{https://wfhmap.com/}.} For each job posting $i$, they impute a binary variable $\hat \theta_i$ representing whether or not the job offers remote work. The variable is generated using DistilBERT, a LLM, applied to metadata on occupation, firm location, job title, posted wage, and a textual job description. As in \citetalias{battagliaInferenceRegressionVariables2025}, we regress log posted wage on $\hat \theta_i$, with and without occupation (SOC2) and full-time/part-time fixed effects (FEs). We focus on a sample of $n = 16{,}315$ job postings from 2022-2023 for the accommodation and food services sector in San Diego, CA. This is a near-singular design: very few jobs in this industry offer remote work ($\hat \pi = 0.024$), making bias correction more challenging. We take $\hat F_+ = 0.009$, as estimated by \citetalias{battagliaInferenceRegressionVariables2025} using $m = 1{,}000$ postings. Since there is no reason to expect the classifier to systematically over- or under-predict, we set $\hat F_- = 0.009$. As a robustness check, we present a second set of results with $\hat F_- = 0.018$.

\begin{table}[t]
\footnotesize
\centering
\caption{Estimated Coefficient of Remote Work Status and 95\% CIs}
\label{tab:app.1}
\makebox[\textwidth][c]{
\begin{threeparttable}
\begin{tabular}{lcccc}
\toprule
 & \multicolumn{2}{c}{No FE} & \multicolumn{2}{c}{With FE} \\
\cmidrule(lr){2-3} \cmidrule(lr){4-5}
Method & Estimate & 95\% CI & Estimate & 95\% CI \\
\midrule
OLS & 0.649 & [0.600, 0.697] & 0.364 & [0.322, 0.406] \\
BCHS bias-corrected \& var.\ adj.\ & 0.897 & [0.668, 1.126] & 0.521 & [0.366, 0.677] \\
No-label resampling & 0.648 & [0.599, 0.695] & 0.363 & [0.322, 0.408] \\
Fixed-label bootstrap & 0.898 & [0.849, 0.944] & 0.522 & [0.482, 0.563] \\
Coupled-label bootstrap & 0.896 & [0.846, 0.944] & 0.510 & [0.473, 0.549] \\
Coupled-label, rotation \& var. adj. & 0.899 & [0.752, 1.062] & 0.520 & [0.413, 0.643] \\
\bottomrule
\end{tabular}
\begin{tablenotes}
\item \textit{Note:} Point estimates and 95\% confidence intervals for the slope coefficient in a regression of log wages on remote work status, with and without occupation (SOC2) and full-time/part-time fixed effects (FE). Results are presented with $\hat F_+ = \hat F_- = 0.009$.
\end{tablenotes}
\end{threeparttable}}
\end{table}

\begin{table}[t]
\footnotesize
\centering
\caption{Estimated Coefficient of Remote Work Status and 95\% CIs}
\label{tab:app.2}
\makebox[\textwidth][c]{
\begin{threeparttable}
\begin{tabular}{lcccc}
\toprule
 & \multicolumn{2}{c}{No FE} & \multicolumn{2}{c}{With FE} \\
\cmidrule(lr){2-3} \cmidrule(lr){4-5}
Method & Estimate & 95\% CI & Estimate & 95\% CI \\
\midrule
OLS & 0.649 & [0.600, 0.697] & 0.364 & [0.322, 0.406] \\
BCHS bias-corrected \& var.\ adj.\ & 0.903 & [0.673, 1.134] & 0.525 & [0.368, 0.682] \\
No-label resampling & 0.648 & [0.599, 0.695] & 0.364 & [0.321, 0.410] \\
Fixed-label bootstrap & 1.048 & [0.986, 1.108] & 0.603 & [0.556, 0.647] \\
Coupled-label bootstrap & 1.047 & [0.984, 1.107] & 0.591 & [0.546, 0.638] \\
Coupled-label, rotation \& var. adj. & 0.905 & [0.762, 1.068] & 0.519 & [0.418, 0.640] \\
\bottomrule
\end{tabular}
\begin{tablenotes}[flushleft]
\setlength{\labelsep}{0pt}
\item \textit{Note:} See notes to Table~\ref{tab:app.1}. Results are presented with $\hat F_+ = 0.009$, $\hat F_- = 0.018$.
\end{tablenotes}
\end{threeparttable}}
\end{table}

The first set of results with $\hat F_+ = \hat F_- = 0.009$ is reported in Table~\ref{tab:app.1}. We present OLS estimates and 95\% confidence intervals for the coefficient of $\theta_i$, alongside results based on the analytical correction in \citetalias{battagliaInferenceRegressionVariables2025}, and bias-corrected estimates and confidence intervals using the naive no-label resampling bootstrap, the fixed-label bootstrap, and the coupled-label bootstrap with and without modifications. OLS estimates are around 30\% smaller than those using the analytical correction. As expected, the naive no-label resampling method produces estimates and CIs that are near identical to those produced by OLS. In the baseline specification without fixed effects, both the fixed-label and coupled-label bootstrap are valid. The estimates and CIs they produce are very similar, as predicted. Our preferred coupled-label bootstrap with variance correction and rotation produces estimates and CIs that are close to those produced by the analytical correction. Similar findings are obtained when FEs are included.

As a robustness check, Table~\ref{tab:app.2} presents results with $\hat F_-$ doubled. The estimates and CIs are largely similar to those presented in Table~\ref{tab:app.1}, except for the fixed-label and coupled-label bootstrap (without rotation). For those methods, estimates and CIs are shifted to the right, well above those obtained using analytical methods. This illustrates the importance of using the rotation in near-singular designs.

\section{Conclusion}

In this paper, we study whether the bootstrap can correct bias and deliver valid inference in regressions with binary labels that have been generated by AI or ML methods. We show that a seemingly natural \emph{fixed-label bootstrap}, which generates data using the estimated labels while relying on a corrupted version in estimation, is generally invalid unless a strong independence condition between the latent true labels and other covariates holds. We propose a \emph{coupled-label bootstrap} that jointly resamples the pair of true and imputed labels and show it is valid without such independence restrictions. We also introduce two further finite-sample modifications: a variance correction and a Hessian rotation for near-singular designs. Our recommendation for empirical practice is to use the coupled-label bootstrap with both of these modifications.

\bibliography{ub}

\appendix\setcounter{section}{0}
\section{Appendix}
\numberwithin{lemma}{section}
\renewcommand{\thelemma}{\Alph{section}.\arabic{lemma}}
\setcounter{lemma}{0}

\subsection{Proofs for Section~\ref{sec:general}}

\begin{proof}[Proof of Theorem~\ref{Theorem:generalbootstraptheory}]We can write
\begin{flalign*} 
Y^*_i=\hat{\beta}^\prime X^*_i+u^*_i=\hat{\beta}^\prime \hat{X}^*_i-\hat{\beta}^\prime (\hat{X}^*_i-X^*_i) +u^*_i,
\end{flalign*}from which we obtain $\sqrt{n}(\hat{\beta}^*-\hat{\beta})=\mathcal{I}^*_{1n}+\mathcal{I}^*_{2n}$, where
\begin{flalign*}\mathcal{I}^*_{1n}&=\Big(n^{-1}\sum_{i=1}^{n}\hat{X}^*_i\hat{X}^{*\prime}_i\Big)^{-1}n^{-1/2}\sum_{i=1}^{n}\hat{X}^*_iu^*_i ,\\
\mathcal{I}^*_{2n}&=-\Big(n^{-1}\sum_{i=1}^{n}\hat{X}^*_i\hat{X}^{*\prime}_i\Big)^{-1}n^{-1/2}\sum_{i=1}^{n}\hat{X}^*_i(\hat{X}^*_i-X^*_i)^\prime \hat{\beta}.
\end{flalign*}
To study the first term, let $\hat{X}^*_i=\hat{X}_i+(X^*_i-\hat{X}_i)+(\hat{X}^*_i-X^*_i)$, and write
\[
  n^{-1}\sum_{i=1}^{n}\hat{X}^*_i\hat{X}^{*\prime}_i = n^{-1}\sum_{i=1}^{n}\hat{X}_i\hat{X}^{\prime}_i+\mathcal{A}^*_{1n}+\mathcal{B}^*_{1n}+\mathcal{C}^*_{1n}+\mathcal{A}^*_{2n}+\mathcal{B}^*_{2n}+\mathcal{C}^*_{2n}+\mathcal{A}^*_{3n} +\mathcal{B}^*_{3n},
\]
where
\begin{align*}
\mathcal{A}^*_{1n}&=\frac{1}{n}\sum_{i=1}^{n}\hat{X}_i(X^*_i-\hat{X}_i)^\prime, &
\mathcal{B}^*_{1n}&=(\mathcal{A}^*_{1n})^\prime, & 
\mathcal{C}^*_{1n}&=\frac{1}{n}\sum_{i=1}^{n}(X^*_i-\hat{X}_i)(X^*_i-\hat{X}_i)^\prime,\\
\mathcal{A}^*_{2n}&=\frac{1}{n}\sum_{i=1}^{n}\hat{X}_i(\hat{X}^*_i-X^*_i)^\prime, & 
\mathcal{B}^*_{2n}&=(\mathcal{A}^*_{2n})^\prime, & 
\mathcal{C}^*_{2n}&=\frac{1}{n}\sum_{i=1}^{n}(\hat{X}^*_i-X^*_i)(\hat{X}^*_i-X^*_i)^\prime ,
\end{align*}
and
\begin{align*}
\mathcal{A}^*_{3n}&=\frac{1}{n}\sum_{i=1}^{n}(X_i^* - \hat{X}_i)(\hat{X}_i^*-X_i^*)^\prime, &
\mathcal{B}^*_{3n}&= (\mathcal{A}^*_{3n})^\prime.
\end{align*}
We can show that $n^{-1}\sum_{i=1}^{n}\hat{X}_i\hat{X}^{\prime}_i$ converges in probability to $Q$ by Assumptions~\ref{a1suff1}\ref{a1suff1.1} and~\ref{a1suff2}\ref{a1suff2.1}, whereas the remainder terms converge to 0 by applying Cauchy--Schwarz inequality and Assumption~\ref{b1}\ref{b1.1}. Hence, $n^{-1}\sum_{i=1}^{n}\hat{X}^*_i\hat{X}^{*\prime}_i\rightarrow_{p^*} Q.$ Since
$$n^{-1/2}\sum_{i=1}^{n}\hat{X}^*_iu^*_i=n^{-1/2}\sum_{i=1}^{n}\hat{X}_iu^*_i+n^{-1/2}\sum_{i=1}^{n}\big(\hat{X}^*_i-X^*_i\big)u^*_i+n^{-1/2}\sum_{i=1}^{n}\big(X^*_i-\hat{X}_i\big)u^*_i,$$ where the last two terms converge to zero by Assumption~\ref{b1}\ref{b1.3}, we can then use Assumption~\ref{b2} to obtain $\mathcal{I}^*_{1n}\rightarrow_{d^*}N(0,V)$, with $V=Q^{-1}\Sigma Q^{-1}$.

For the second term, we can use Assumption~\ref{b1}\ref{b1.2} to obtain $\mathcal{I}^*_{2n}\rightarrow_{p^*} b$ with $b=-Q^{-1}B\beta$.
\end{proof}

\subsection{Proofs for Section~\ref{sec:labels}}

\begin{proof}[Proof of Lemma~\ref{l1}]
First, Assumption~\ref{a1suff2} holds under Assumptions~\ref{a2}\ref{a2.1}--\ref{a2.2}. 
To verify Assumption~\ref{a1suff1}\ref{a1suff1.1}, we first note by Assumption~\ref{a2}\ref{a2.4} that 
\begin{equation}\label{eq:delta}
 \E[ \mathbb{I}[\hat \theta_i \neq \theta_i]] = \E[\hat \theta_i(1-\theta_i)] + \E[\theta_i(1-\hat \theta_i)] \to 0 .
\end{equation}
We now use Markov's inequality, noting that 
\[
 \begin{aligned}
 \E[\|\hat X_i - X_i\|^2] & = \E[\|g(1, Z_i) - g(0,Z_i)\|^2 \mathbb{I}[\hat \theta_i \neq \theta_i]] \\
 & \leq 4 \E[G_i^2 \mathbb{I}[\hat \theta_i \neq \theta_i]] \\
 & \leq 4 \E[G_i^4]^{\frac 12} \E[ \mathbb{I}[\hat \theta_i \neq \theta_i]]^{\frac 12} \\
 & \to 0 ,
 \end{aligned}
\]
where the first equality is by definition, the second and third lines are by the triangle and Cauchy--Schwarz inequalities, and the final line is by Assumption~\ref{a2}\ref{a2.2} and~\eqref{eq:delta}.
For Assumption~\ref{a1suff1}\ref{a1suff1.3}, first note by Assumption~\ref{a2}\ref{a2.3} that $\sqrt n \E[(\hat X_i - X_i) u_i] \to 0$. Then
\[
 \begin{aligned}
 & \E\left[ \left\| \frac{1}{\sqrt n} \sum_{i=1}^n \left( (\hat X_i - X_i) u_i - \E[(\hat X_i - X_i) u_i] \right) \right\|^2 \right] \\
 & \leq \E\left[ \|(\hat X_i - X_i) u_i \|^2 \right] \\
 & = \E\left[ \|g(1, Z_i) - g(0,Z_i)\|^2 u_i^2 \mathbb{I}[\hat \theta_i \neq \theta_i] \right] \\
 & \leq 4 \E[G_i^2 u_i^2 \mathbb{I}[\hat \theta_i \neq \theta_i]] \\
 & \leq 4 \E[G_i^{4+\delta}]^{\frac{2}{4+\delta}} \E[|u_i|^{4+\delta}]^{\frac{2}{4+\delta}} \E[ \mathbb{I}[\hat \theta_i \neq \theta_i]]^{\frac{\delta}{4+\delta}} \to 0 ,
 \end{aligned}
\]
again by Assumption~\ref{a2}\ref{a2.1}--\ref{a2.2} and \eqref{eq:delta}.

It remains to verify Assumption~\ref{a1suff1}\ref{a1suff1.2}. To this end, we first write
\[
 \mathcal I_n = \frac{1}{\sqrt n} \sum_{i=1}^n \hat X_i (\hat X_i - X_i)' 
 = \frac{1}{\sqrt n} \sum_{i=1}^n \left( \hat \theta_i(1 - \theta_i) D_{+i} + \theta_i(1-\hat \theta_i) D_{-i} \right) .
\]
Taking expectations, we have by Assumption~\ref{a2}\ref{a2.4}--\ref{a2.5} that
\[
 \begin{aligned}
 \E[\mathcal I_n] & = \sqrt n \E [ \hat \theta_i(1 - \theta_i) D_{+i} + \theta_i(1-\hat \theta_i) D_{-i} ] \\
  & \to \kappa_+ \E[D_{+i}] + \kappa_- \E[D_{-i}] .
 \end{aligned}
\]
Finally, with $\|\cdot\|_F$ denoting the Frobenius norm, we have
\[
 \begin{aligned}
 \E\left[ \left\| \mathcal I_n - \E[\mathcal I_n] \right\|_F^2 \right] 
 & \leq \E \left[ \left\|  D_{+i} \right\|_F^2 \mathbb{I}[\hat \theta_i = 1, \theta_i = 0] +  \left\|  D_{-i} \right\|_F^2 \mathbb{I}[\theta_i = 1, \hat \theta_i = 0] \right] \\
 & \leq \E [ G_i^4 \mathbb{I}[\hat \theta_i \neq \theta_i ] ] \\
 & \leq \E [ G_i^{4+\delta} ]^{\frac{4}{4+\delta}} \E [ \mathbb I[\hat \theta_i \neq \theta_i] ]^{\frac{\delta}{4+\delta}} \to 0 ,
 \end{aligned}
\]
by Assumption~\ref{a2}\ref{a2.1}--\ref{a2.2} and \ref{a2.4}. Hence, $\mathcal I_n \to_p \kappa_+ \E[D_{+i}] + \kappa_- \E[D_{-i}]$.
\end{proof}

To prove Lemma~\ref{lemma_boot_labels} we rely on Lemmas~\ref{lemma_bootCLT}--\ref{lemma_B1(ii)} below. 
Lemma~\ref{lemma_bootCLT} verifies Assumption~\ref{b2} for the wild bootstrap. Lemma~\ref{lemma_B1(i)(iii)} verifies Assumptions~\ref{b1}\ref{b1.1} and \ref{b1.3}, and Lemma~\ref{lemma_B1(ii)} verifies Assumption~\ref{b1}\ref{b1.2}.
For Lemma~\ref{lemma_bootCLT}, Assumption~\ref{a2} and standard conditions on the wild bootstrap suffice. No assumptions on $\hat{\theta}^*_i$ are required.

\begin{lemma}\label{lemma_bootCLT} 
Suppose that Assumption~\ref{a2} holds. Let $u^*_i=\hat{u}_i\eta_i$ where $\eta_i$ are i.i.d.~$(0,1)$ such that $\E^*[|\eta_i|^{2+\delta}]<C<\infty$ for $\delta>0$. Then Assumption~\ref{b2} holds.
\end{lemma}

\begin{proof}[Proof of Lemma~\ref{lemma_bootCLT}]We can write
\begin{flalign}\label{eq.bootCLT1}
  \frac{1}{\sqrt{n}}\sum_{i=1}^{n}\hat{X}_iu^*_i=\frac{1}{\sqrt{n}}\sum_{i=1}^{n}X_iu^*_i+ \frac{1}{\sqrt{n}}\sum_{i=1}^{n}(\hat{X}_i-X_i)u^*_i= \mathcal{Z}^*_{1n}+\mathcal{Z}^*_{2n}.
\end{flalign}
We will show (i) $\mathcal{Z}^*_{1n}\to_{d^*}N(0,\Sigma)$, and (ii) $\mathcal{Z}^*_{2n}\to_{p^*}0$. Then Assumption~\ref{b2} follows from (i) and (ii). We start with (i). Write
\begin{flalign}\label{eq.bootCLT2}
\frac{1}{\sqrt{n}}\sum_{i=1}^{n}X_i u^*_i=\frac{1}{\sqrt{n}}\sum_{i=1}^{n}X_i u_i\eta_i+\frac{1}{\sqrt{n}}\sum_{i=1}^{n}X_i(\hat{u}_i-u_i)\eta_i.
\end{flalign}
Next we show that the first term in (\ref{eq.bootCLT2}) follows a bootstrap CLT. Note that $\E^*[X_i u_i\eta_i]=0$ and $\V^*(n^{-1/2}\sum_{i=1}^n X_i u_i\eta_i)=n^{-1}\sum_{i=1}^{n}X_iX_i'u^2_i\to_p \Sigma\equiv \E[X_iX_i'u^2_i]$ under Assumption~\ref{a2}\ref{a2.1}--\ref{a2.2}. Since, conditional on the original sample, the draws $X_iu_i\eta_i$ are independent but non-identically distributed, we verify Lyapunov's condition using the Cramer--Wold device. Let $\alpha\in \mathbb{R}^k$ such that $\alpha'\alpha=1$. We show that 
\begin{equation}\label{eq.bootCLT2.a}
 n^{-(1+\delta/2)}\sum_{i=1}^{n}\E^*[|\alpha' X_i u_i \eta_i|^{2+\delta}]\to_p 0
\end{equation}
for $\delta>0$. Since $\E^*[|\alpha' X_i u_i \eta_i|^{2+\delta}]\le \|X_i\|^{2+\delta}|u_i|^{2+\delta}\E^*[|\eta_i|^{2+\delta}]$, where $\E^*[|\eta_i|^{2+\delta}]<C$ by assumption, an application of the Cauchy--Schwarz inequality yields 
\begin{flalign*}
\frac{1}{n^{1+\delta/2}}\sum_{i=1}^{n}\E^*[|\alpha' X_i u_i \eta_i|^{2+\delta}]&\le 
\frac{C}{n^{\delta/2}}\Big(\frac{1}{n}\sum_{i=1}^{n}\|X_i\|^{2(2+\delta)}\Big)^{1/2}\Big(\frac{1}{n}\sum_{i=1}^{n}|u_i|^{2(2+\delta)}\Big)^{1/2}\\
&=O_p(n^{-\delta/2}),
\end{flalign*}
provided $\E[\|X_i\|^{2(2+\delta)}]<\infty$ and $\E[|u_i|^{2(2+\delta)}]<\infty$, which is implied by Assumption~\ref{a2}\ref{a2.1}--\ref{a2.2} with a suitable redefinition of $\delta$ above. This proves (\ref{eq.bootCLT2.a}).

Next, we show that the second term in (\ref{eq.bootCLT2}) is asymptotically negligible, i.e., $n^{-1/2}\sum_{i=1}^{n}X_i(\hat{u}_i-u_i)\eta_i\to_{p^*}0$. Since $\E^*[\eta_i]=0$, this term is centered at zero, and it suffices to prove that its bootstrap variance converges to zero in probability. Since $\E^*[\eta_i^2]=1$, we have that
\begin{multline}\label{eq.bootCLT2.b}
\V^*\bigg(n^{-1/2}\sum_{i=1}^{n}X_i(\hat{u}_i-u_i)\eta_i\bigg)=\frac{1}{n}\sum_{i=1}^{n}X_iX_i'(\hat{u}_i-u_i)^2\\
\le \frac{3}{n}\sum_{i=1}^{n}X_iX_i'[X_i'(\hat{\beta}-\beta)]^2+\frac{3}{n}\sum_{i=1}^{n}X_iX_i'[(\hat{X}_i-X_i)'\beta]^2 \\ 
+\frac{3}{n}\sum_{i=1}^{n}X_iX_i'[(\hat{X}_i-X_i)'(\hat{\beta}-\beta)]^2 
= 3(\mathcal{A}_{1n}+\mathcal{A}_{2n}+\mathcal{A}_{3n}).
\end{multline}
We can bound the Frobenius norm of $\mathcal{A}_{1n}$ as
\[
 \bigg\|\frac{1}{n}\sum_{i=1}^{n}X_iX_i'[X_i'(\hat{\beta}-\beta)]^2\bigg\|\le \frac{1}{n}\sum_{i=1}^{n}\|X_iX_i'\||X_i'(\hat{\beta}-\beta)|^2
\le \frac{1}{n}\sum_{i=1}^{n}\|X_i\|^4\|\hat{\beta}-\beta\|^2=o_p(1),
\]since $\E[\|X_i\|^4]<\infty$ and $\|\hat{\beta}-\beta\|^2= O_p(n^{-1})=o_p(1)$ under Assumption~\ref{a2}. To prove that $\mathcal{A}_{2n}=o_p(1)$, we can use Cauchy--Schwarz inequality to bound its Frobenius norm as 
\begin{multline*}
\bigg\|\frac{1}{n}\sum_{i=1}^{n}X_iX_i'[(\hat{X}_i-X_i)'\beta]^2\bigg\|\le \frac{1}{n}\sum_{i=1}^{n}\|X_i\|^2\|\hat{X}_i-X_i\|^2\|\beta\|^2\\
\le \Big(\frac{1}{n}\sum_{i=1}^{n}\|X_i\|^4\Big)^{1/2}\Big(\frac{1}{n}\sum_{i=1}^{n}\|\hat{X}_i-X_i\|^4\Big)^{1/2}\|\beta\|^2,
\end{multline*}
where the first factor is $O_p(1)$ if $\E[\|X_i\|^4]<\infty$ and the second term is $o_p(1)$ under Assumption~\ref{a2}. To see this, note that $\hat{X}_i - X_i =(\hat{\theta}_i-\theta_i)(g(1,Z_i) - g(0,Z_i))$. By H\"{o}lder's inequality with $p,q>1$ such that $p^{-1}+q^{-1}=1$,
\begin{flalign*}
\frac{1}{n}\sum_{i=1}^{n}\|\hat{X}_i-X_i\|^4&\le\frac{1}{n}\sum_{i=1}^{n}\I[\hat{\theta}_i\ne\theta_i]\|g(1,Z_i)-g(0,Z_i)\|^4\\
  &\le \Big(\frac{1}{n}\sum_{i=1}^{n}\I[\hat{\theta}_i\ne \theta_i]\Big)^{1/p}\Big(\frac{1}{n}\sum_{i=1}^{n}\|g(1,Z_i)-g(0,Z_i)\|^{4q}\Big)^{1/q}\\
  &\le  \Big(\frac{1}{n}\sum_{i=1}^{n}\I[\hat{\theta}_i\ne \theta_i]\Big)^{\delta/(4+\delta)}\Big(\frac{1}{n}\sum_{i=1}^{n}\|g(1,Z_i)-g(0,Z_i)\|^{4+\delta}\Big)^{4/(4+\delta)},
\end{flalign*}setting $q=1+\delta/4$ and $p=(4+\delta)/\delta$. Under Assumption~\ref{a2}\ref{a2.2},  $\E[G_i^{4+\delta}]<\infty$, which implies that the second factor is $O_p(1)$, whereas the first factor is $o_p(1)$ by Assumption~\ref{a2}\ref{a2.4}. The proof that $\mathcal{A}_{3n}=o_p(1)$ follows similarly.

We end the proof of the lemma by showing that $\mathcal{Z}^*_{2n}\to_{p^*}0$. Since $\E^*[u^*_i]=0$, we have $\E^*[\mathcal{Z}^*_{2n}]=0$ and the result follows by showing that $\V^*(\mathcal{Z}^*_{2n})\to_p 0$. Using the fact that $u^*_i=\hat{u}_i\eta_i$ with $\eta_i$ i.i.d.~$(0,1)$ yields
\[
\V^*(\mathcal{Z}^*_{2n})=\frac{1}{n}\sum_{i=1}^{n}(\hat{X}_i-X_i)(\hat{X}_i-X_i)'\hat{u}^2_i,
\]whose Frobenius norm is bounded by
\[\frac{1}{n}\sum_{i=1}^{n}\|\hat{X}_i-X_i\|^2|\hat{u}_i|^2\le \Big(\frac{1}{n}\sum_{i=1}^{n}\|\hat{X}_i-X_i\|^4\Big)^{1/2}\Big(\frac{1}{n}\sum_{i=1}^{n}|\hat{u}_i|^4\Big)^{1/2}.
\]
We have shown above that the first factor is $o_p(1)$ under Assumption~\ref{a2}. For the second, note $|\hat u_i|^4 \leq 8 |u_i|^4 + 8 |u_i - \hat u_i|^4$, where $\frac 1n \sum_{i=1}^n |u_i|^4 = O_p(1)$ by Assumption~\ref{a2}\ref{a2.1}, and $\frac 1n \sum_{i=1}^n |\hat u_i - u_i|^4 = o_p(1)$ can be shown by similar arguments to those used in (\ref{eq.bootCLT2.b}).
\end{proof}

Lemma~\ref{lemma_B1(i)(iii)} requires additional conditions. In particular, we assume that $u^*_i$ is obtained by the wild bootstrap independently of $(\theta^*_i,\hat{\theta}^*_i)$, conditional on the data, and impose the following high-level condition on the bootstrap labels $(\theta_i^*, \hat \theta_i^*)$.
\setcounter{assumption}{0}
\renewcommand{\theassumption}{L.\arabic{assumption}}
\begin{assumption}\label{bL.1}  
\begin{enumerate}[label={(\roman*)}]
   \item\label{bL.11}$\frac{1}{n}\sum_{i=1}^{n}\E^*[\I[\hat{\theta}^*_i\ne \theta^*_i]]\to_p 0$.
   \item\label{bL.12} $\frac{1}{n}\sum_{i=1}^{n}\E^*[\I[\theta^*_i\ne \hat{\theta}_i]]\to_p 0$.
\end{enumerate}
\end{assumption}
Assumption~\ref{bL.1}\ref{bL.11} requires the overall misclassification rate of the bootstrap labels to converge to zero in probability. This condition is implied by Assumption~\ref{ass:L}\ref{ass:L1}. Assumption~\ref{bL.1}\ref{bL.12} reproduces Assumption~\ref{ass:L}\ref{ass:L0}. 

\begin{lemma}\label{lemma_B1(i)(iii)} 
Suppose that Assumption~\ref{a2} holds. Let $u^*_i=\hat{u}_i\eta_i$ where $\eta_i$ are i.i.d.~$(0,1)$ independently of $(\theta^*_i,\hat{\theta}^*_i)$ and assume that $\theta^*_i$ and $\hat{\theta}^*_i$ satisfy Assumption~\ref{bL.1}. Then Assumptions~\ref{b1}\ref{b1.1} and \ref{b1}\ref{b1.3} hold.
\end{lemma}

\begin{proof}[Proof of Lemma~\ref{lemma_B1(i)(iii)}] The first part of Assumption~\ref{b1}\ref{b1.1} follows by noting that 
\begin{flalign*}
\frac{1}{n}\sum_{i=1}^{n}\|\hat{X}^*_i-X^*_i\|^2&=\frac{1}{n}\sum_{i=1}^{n}\I[\hat{\theta}^*_i\ne\theta^*_i]\|g(1,Z_i)-g(0,Z_i)\|^2\\
&\le \Big(\frac{1}{n}\sum_{i=1}^{n}\I[\hat{\theta}^*_i\ne\theta^*_i]\Big)^{1/2} \Big(\frac{1}{n}\sum_{i=1}^{n}\|g(1,Z_i)-g(0,Z_i)\|^4\Big)^{1/2},
\end{flalign*}which is $o_{p^*}(1)$ provided the first factor is $o_{p^*}(1)$ (since the second factor is $O_p(1)$ under Assumption~\ref{a2}\ref{a2.2}). Assumption~\ref{bL.1}\ref{bL.11} suffices for this given Markov's inequality, noting that $\I(\hat{\theta}^*_i\ne \theta^*_i)$ is non-negative. An analogous argument, in conjunction with Assumption~\ref{bL.1}\ref{bL.12}, verifies the second part of Assumption~\ref{b1}\ref{b1.1}. 

Next, we verify Assumption~\ref{b1}\ref{b1.3}. We focus on showing the first part, since the second part follows by an analogous argument. Note that $\E^*[(\hat{X}^*_i-X^*_i)u^*_i]=\E^*[(\hat{X}^*_i-X^*_i)]\E^*[u^*_i]=0$ since $u^*_i = \hat u_i \eta_i$, where $\eta_i$ is independent of $\hat u_i$ and $\hat{X}^*_i-X^*_i$, conditional on the sample, and $\E^*[\eta_i]=0$. Hence, it suffices to show that $\V^*(n^{-1/2}\sum_{i=1}^{n}(\hat{X}^*_i-X^*_i)u^*_i)\to_p 0$. To verify this condition, we can use the independence between $\eta_i$ and $\hat{X}^*_j-X^*_j$ for all $i,j$ and the fact that $u^*_i$ is independent across $i$ (by the wild bootstrap property). These assumptions imply that
\begin{flalign*}
  & \V^*\Big(\frac{1}{\sqrt{n}}\sum_{i=1}^{n}(\hat{X}^*_i-X^*_i)u^*_i\Big)\\
  &=\frac{1}{n}\sum_{i=1}^{n}\sum_{j=1}^{n}\E^*[(\hat{X}^*_i-X^*_i)u^*_i(\hat{X}^*_j-X^*_j)'u^*_j] \\
  &=\frac{1}{n}\sum_{i=1}^{n}\E^*[(\hat{X}^*_i-X^*_i)(\hat{X}^*_i-X^*_i)']\E^*[u^{*2}_i]\\
  &=\frac{1}{n}\sum_{i=1}^{n}\E^*[(\hat{\theta}^*_i-\theta^*_i)^2](g(1,Z_i)-g(0,Z_i))(g(1,Z_i)-g(0,Z_i))'\hat{u}^2_i
\end{flalign*}
where the last equality follows because $\hat{X}^*_i-X^*_i=(\hat{\theta}^*_i-\theta^*_i)(g(1,Z_i)-g(0,Z_i))$. Noting that $\E^*[(\hat{\theta}^*_i-\theta^*_i)^2]=\E^*[\I[\hat{\theta}^*_i\ne \theta^*_i]]$, we can bound the Frobenius norm of the bootstrap variance above by
\begin{flalign*}
&\frac{1}{n}\sum_{i=1}^{n}\E^*[\I[\hat{\theta}^*_i\ne\theta^*_i]]\|g(1,Z_i)-g(0,Z_i)\|^2|\hat{u}^2_i|\\
&\le \Big(\frac{1}{n}\sum_{i=1}^{n}\big(\E^*[\I[\hat{\theta}^*_i\ne\theta^*_i]]\big)^p\Big)^{1/p}\Big(\frac{1}{n}\sum_{i=1}^{n}\big(\|g(1,Z_i)-g(0,Z_i)\|^{2q}|\hat{u}_i|^{2q}\big)\Big)^{1/q}
\end{flalign*} for any $p,q>1$ such that $p^{-1}+q^{-1}=1$. Note $\big(\E^*[\I[\hat{\theta}^*_i\ne\theta^*_i]]\big)^p\leq \E^*[\I[\hat{\theta}^*_i\ne\theta^*_i]]$ by Jensen's inequality. Taking $q=1+\delta/4$ for $\delta>0$ as in Assumption~\ref{a1suff1} implies $p^{-1}=\delta/(4+\delta)$, yielding the following bound:
\begin{flalign*}
& \Big(\frac{1}{n}\sum_{i=1}^{n}\E^*[\I[\hat{\theta}^*_i\ne\theta^*_i]]\Big)^{\delta/(4+\delta)}\Big(\frac{1}{n}\sum_{i=1}^{n}\|g(1,Z_i)-g(0,Z_i)\|^{2+\delta/2}|\hat{u}_i|^{2+\delta/2}\Big)^{4/(4+\delta)}.
\end{flalign*}
Assumption~\ref{bL.1}\ref{bL.11} implies that the first factor is $o_{p^*}(1)$. It remains to show the second factor is $O_p(1)$. By the Cauchy--Schwarz inequality, 
\begin{multline*}
\frac{1}{n}\sum_{i=1}^{n}\|g(1,Z_i)-g(0,Z_i)\|^{2+\delta/2}|\hat{u}_i|^{2+\delta/2} \\
\le \Big(\frac{1}{n}\sum_{i=1}^{n}\|g(1,Z_i)-g(0,Z_i)\|^{4+\delta} \Big)^{1/2}\Big(\frac{1}{n}\sum_{i=1}^{n}|\hat{u}_i|^{4+\delta}\Big)^{1/2},
\end{multline*}
where the first factor on the right-hand side is $O_p(1)$ by Assumption~\ref{a2}\ref{a2.2}. The second factor can be shown to be $O_p(1)$, noting $|\hat u_i|^{4+\delta} \leq 2^{3 + \delta}(|u_i|^{4+\delta} + |\hat u_i - u_i|^{4+\delta})$, where $\frac 1n \sum_{i=1}^n |u_i|^{4+\delta} = O_p(1)$ follows by Assumption~\ref{a2}\ref{a2.1}, and $\frac 1n \sum_{i=1}^n |\hat u_i - u_i|^{4+\delta} = O_p(1)$ follows by similar arguments to those used in (\ref{eq.bootCLT2.b}).
\end{proof}

Finally, Lemma~\ref{lemma_B1(ii)} verifies Assumption~\ref{b1}\ref{b1.2} under Assumption~\ref{ass:L}\ref{ass:L1}--\ref{ass:L3}.

\begin{lemma}\label{lemma_B1(ii)} 
Suppose that Assumption~\ref{a2} holds. If $\theta^*_i$ and $\hat{\theta}^*_i$ satisfy Assumption~\ref{ass:L}\ref{ass:L1}--\ref{ass:L3}, then Assumption~\ref{b1}\ref{b1.2} holds.
\end{lemma}

\begin{proof}[Proof of Lemma~\ref{lemma_B1(ii)}]Since $\hat{X}^*_i-X^*_i=(\hat{\theta}^*_i-\theta^*_i)[g(1,Z_i)-g(0,Z_i)]$, we can write
\begin{flalign*}
\hat{X}^*_i(\hat{X}^*_i - X^*_i )'=\hat{\theta}^*_i(1-\theta^*_i)D_{+i} +\theta^*_i(1-\hat{\theta}^*_i)D_{-i}.\label{eq.x*(x*-x)}
\end{flalign*}Assumption~\ref{b1}\ref{b1.2} can be rewritten as
\[
B^*_n\equiv \frac{1}{\sqrt{n}}\sum_{i=1}^{n}(\hat{\theta}^*_i(1-\theta^*_i)D_{+i}+ \theta^*_i(1-\hat{\theta}^*_i)D_{-i})\to_{p^*} \kappa_+ \E[D_{+i}] +\kappa_{-}\E[D_{-i}]\equiv B,
\]
with $B$ given in Lemma~\ref{l1}. Adding and subtracting yields $B^*_n=B^*_{1n}+B^*_{2n}$, where
\begin{flalign*}
 B^*_{1n}&=\frac{1}{\sqrt{n}}\sum_{i=1}^{n}\hat{\theta}^*_i(1-\theta^*_i)\E[D_{+i}]+ \theta^*_i(1-\hat{\theta}^*_i)\E[D_{-i}], 
 \\
 B^*_{2n}&=\frac{1}{\sqrt{n}}\sum_{i=1}^{n}\hat{\theta}^*_i(1-\theta^*_i)\big(D_{+i}-\E[D_{+i}]\big)+ \theta^*_i(1-\hat{\theta}^*_i)\big(D_{-i}-\E[D_{-i}]\big). 
\end{flalign*}
Assumption~\ref{ass:L}\ref{ass:L3} implies that $B^*_{2n}\to_{p^*} 0$. Next, we show that $B^*_{1n}\to_{p^*} B$. Since $\E[D_{+i}]$ and $\E[D_{-i}]$ are independent of $i$, the result follows by noting that Assumptions~\ref{ass:L}\ref{ass:L1}--\ref{ass:L2} imply that $\frac{1}{\sqrt{n}}\sum_{i=1}^{n}\hat{\theta}^*_i(1-\theta^*_i)\to_{p^*} \kappa_+$ and $\frac{1}{\sqrt{n}}\sum_{i=1}^{n}\theta^*_i(1-\hat{\theta}^*_i)\to_{p^*} \kappa_-$, by Chebyshev's inequality.\end{proof}

\begin{proof}[Proof of Lemma~\ref{lemma_boot_labels}] 
This follows from Lemmas~\ref{lemma_bootCLT}--\ref{lemma_B1(ii)} given our assumptions on $u^*_i$ and Assumption~\ref{ass:L}, and noting that Assumption~\ref{ass:L}\ref{ass:L1} implies Assumption~\ref{bL.1}\ref{bL.11}.
\end{proof}

\begin{proof}[Proof of Theorem~\ref{TheoremCoupledB}]
This is a special case of the proof of Theorem~\ref{TheoremCoupledB.variance} below, setting $\hat F^*_+ = \hat F_+$ and $\hat F^*_- = \hat F_-$.
\end{proof}

\begin{proof}[Proof of Theorem~\ref{TheoremCoupledB.variance}]
In view of Theorem~\ref{Theorem:generalbootstraptheory} and Lemma~\ref{lemma_boot_labels}, it is enough to show that Assumption~\ref{ass:L} holds.  

Assumption~\ref{ass:L}\ref{ass:L0} requires $\frac{1}{n}\sum_{i=1}^n \E^*[\I[\theta_i^* \neq \hat{\theta}_i]] \to_p 0$. From the DGP for $(\theta_i^*, \hat{\theta}_i^*)$ in (\ref{eq:coupled.1.var}) and (\ref{eq:coupled.2.var}), we see that if $\hat{\theta}_i = 1$ then $\theta_i^* = 0$ with probability $\hat{F}_+ + \hat{F}_- \frac{1-\hat \pi}{\hat \pi}$ because $\E^*[\hat F^*_+] = \hat F_+$ and $\E^*[\hat F^*_-] = \hat F_-$. Similarly, if $\hat{\theta}_i = 0$ then $\theta_i^* = 1$ with probability $\hat{F}_- + \hat F_+ \frac{\hat \pi}{1- \hat \pi}$. Therefore,
\[
\frac{1}{n}\sum_{i=1}^n \E^*[\I[\theta_i^* \neq \hat{\theta}_i]] = 2\hat{\pi}\hat{F}_+ + 2(1-\hat{\pi})\hat{F}_- = O_p(n^{-1/2}) \to_p 0,
\]
where we used $\sqrt{n}\hat{F}_ +\to_p \kappa_+$ and $\sqrt{n}\hat{F}_ -\to_p \kappa_-$. 
Hence, Assumption~\ref{ass:L}\ref{ass:L0} holds. 

For part~\ref{ass:L1}, consider $\frac{1}{\sqrt{n}}\sum_{i=1}^{n}\E^*[\hat{\theta}^*_i(1-\theta^*_i)]$. As $\hat{\theta}^*_i(1-\theta^*_i)$ are independent Bernoulli random variables with common success probability $\hat F^*_+$ conditional on the data and $\hat F^*_+$, and $\E^*[\hat F^*_+] = \hat F_+$, we have
\[
 \frac{1}{\sqrt{n}}\sum_{i=1}^{n}\E^*[\hat{\theta}^*_i(1-\theta^*_i)] = \sqrt n \E^*[\hat F^*_+] = \sqrt n \hat F_+ \to_p \kappa_+,
\] 
again by $\sqrt{n}\hat{F}_ +\to_p \kappa_+$. The result $\frac{1}{\sqrt{n}}\sum_{i=1}^{n}\E^*[\theta^*_i(1-\hat{\theta}^*_i)]\to_p \kappa_{-}$ follows similarly.

For part~\ref{ass:L2}, consider $\V^*\big(\frac{1}{\sqrt{n}}\sum_{i=1}^{n}\hat{\theta}^*_i(1-\theta^*_i)\big)$. Using the fact that $\hat{\theta}^*_i(1-\theta^*_i)$ are independent Bernoulli random variables conditional on the data and $\hat F^*_+$, we have 
\[
 \V^*\left(\left. \frac{1}{\sqrt{n}}\sum_{i=1}^{n}\hat{\theta}^*_i(1-\theta^*_i)\right| \hat F^*_+\right) = \V^*(\hat{\theta}^*_i(1-\theta^*_i)|\hat F^*_+) = \hat F^*_+(1-\hat F^*_+).
\]
It follows by the fact that $m\hat F^*_+ \sim \text{Binomial}(m,\hat F_+)$, that
\[
 \E^* \left[ \V^*\left(\left. \frac{1}{\sqrt{n}}\sum_{i=1}^{n}\hat{\theta}^*_i(1-\theta^*_i)\right| \hat F^*_+\right) \right] = \hat F_+(1-\hat F_+)\left( 1 - \frac 1m \right) ,
\]
which is $O_p(n^{-1/2})$ because $\sqrt n \hat F_+ \to_p \kappa_+$. Similarly,
\[
 \E^*\left[ \left.  \frac{1}{\sqrt{n}}\sum_{i=1}^{n}\hat{\theta}^*_i(1-\theta^*_i)\right| \hat F^*_+ \right] = \sqrt n \hat F^*_+,
\]
and so 
\[
 \V^* \left( \E^*\left[ \left.  \frac{1}{\sqrt{n}}\sum_{i=1}^{n}\hat{\theta}^*_i(1-\theta^*_i)\right| \hat F^*_+ \right] \right) = \frac{n \hat F_+(1-\hat F_+)}{m} \to_p 0,
\]
because $\sqrt n \hat F_+ \to_p \kappa_+$ and $n/m^2 \to 0$. It follows by the law of total variance that $\V^*\big(\frac{1}{\sqrt{n}}\sum_{i=1}^{n}\hat{\theta}^*_i(1-\theta^*_i)\big) \to_p 0$, as required. The result $\V^*\big(\frac{1}{\sqrt{n}}\sum_{i=1}^{n}\theta^*_i(1-\hat{\theta}^*_i)\big)\to_p 0$ follows similarly.

For part~\ref{ass:L3}, consider $\frac{1}{\sqrt{n}}\sum_{i=1}^{n}\hat{\theta}^*_i(1-\theta^*_i)\big(D_{+i}-\E[D_{+i}]\big)$. By similar arguments to part~\ref{ass:L2}, we have 
\begin{multline*}
 \E^*\left[\frac{1}{\sqrt{n}}\sum_{i=1}^{n}\hat{\theta}^*_i(1-\theta^*_i)\big(D_{+i}-\E[D_{+i}]\big)\right] 
 = \hat F_+ \times \frac{1}{\sqrt n} \sum_{i=1}^n \big(D_{+i}-\E[D_{+i}]\big) \\
 = O_p(n^{-1/2}) \times O_p(1) \to_p 0,
\end{multline*}
where the first factor is because $\sqrt n \hat F_+ \to_p \kappa_+$ and the second is by the CLT and Assumption~\ref{a2}\ref{a2.2}. We now calculate the bootstrap variance of this term. Again using the fact that conditional on $\hat F^*_+$ and the data, the $\hat{\theta}^*_i(1-\theta^*_i)$ are independent Bernoulli random variables, we have
\[
 \E^*\left[\left. \frac{1}{\sqrt{n}}\sum_{i=1}^{n}\hat{\theta}^*_i(1-\theta^*_i)\big(D_{+i}-\E[D_{+i}]\big) \right| \hat F^*_+ \right] 
 = \hat F_+^* \times \frac{1}{\sqrt n} \sum_{i=1}^n \big(D_{+i}-\E[D_{+i}]\big),
\]
and, for notational simplicity, assuming $D_{+i}$ is scalar  (otherwise, the argument applies element-wise): 
\begin{multline*}
 \V^*\left(\left. \frac{1}{\sqrt{n}}\sum_{i=1}^{n}\hat{\theta}^*_i(1-\theta^*_i)\big(D_{+i}-\E[D_{+i}]\big) \right| \hat F^*_+ \right) \\
 = \hat F_+^*(1-\hat F^*_+) \times \frac{1}{n} \sum_{i=1}^n \big(D_{+i}-\E[D_{+i}]\big)^2.
\end{multline*}
Now using the fact that $m \hat F^*_+$ is Binomial$(m,\hat F_+)$ conditional on the data, we have
\begin{multline*}
 \V^* \left( \E^*\left[\left. \frac{1}{\sqrt{n}}\sum_{i=1}^{n}\hat{\theta}^*_i(1-\theta^*_i)\big(D_{+i}-\E[D_{+i}]\big) \right| \hat F^*_+ \right] \right) \\
 = \frac{\hat F_+(1-\hat F_+)}{m} \times \left( \frac{1}{\sqrt n} \sum_{i=1}^n \big(D_{+i}-\E[D_{+i}]\big) \right)^2 = O_p(1/(m \sqrt n)) \times O_p(1) \to_p 0,
\end{multline*}
where the first factor is because $\sqrt n \hat F_+ \to_p \kappa_+$ and the second is by the CLT and Assumption~\ref{a2}\ref{a2.2}. Similarly,
\begin{multline*}
 \E^* \left[ \V^*\left(\left. \frac{1}{\sqrt{n}}\sum_{i=1}^{n}\hat{\theta}^*_i(1-\theta^*_i)\big(D_{+i}-\E[D_{+i}]\big) \right| \hat F^*_+ \right) \right] \\
 = \hat F_+(1-\hat F_+) \left(1 - \frac 1m \right) \times \frac{1}{n} \sum_{i=1}^n \big(D_{+i}-\E[D_{+i}]\big)^2 
 = O_p(n^{-1/2}) \times O_p(1) \to_p 0,
\end{multline*}
where the first factor is because $\sqrt n \hat F_+ \to_p \kappa_+$ and the second is by the LLN and Assumption~\ref{a2}\ref{a2.2}. Therefore, we have shown that both the bootstrap expectation and bootstrap variance of $\frac{1}{\sqrt{n}}\sum_{i=1}^{n}\hat{\theta}^*_i(1-\theta^*_i)\big(D_{+i}-\E[D_{+i}]\big)$ converge in probability to zero. It follows that $\frac{1}{\sqrt{n}}\sum_{i=1}^{n}\hat{\theta}^*_i(1-\theta^*_i)\big(D_{+i}-\E[D_{+i}]\big) \to_{p^*} 0$, as required. The proof that $\frac{1}{\sqrt{n}}\sum_{i=1}^{n}\theta^*_i(1-\hat{\theta}^*_i)\big(D_{-i}-\E[D_{-i}]\big)\to_{p^*} 0$ is analogous.
\end{proof}

\end{document}